\documentclass[aps,twocolumn,showpacs,superscriptaddress,reprint,longbibliography]{revtex4-1}  

\usepackage[colorlinks]{hyperref}
\hypersetup{%
  citecolor=cyan,
  linkcolor=blue,
  urlcolor=cyan
}
\usepackage{graphicx}
\usepackage{bm}
\usepackage{amssymb}
\usepackage{amsmath}
\usepackage{amsthm}
\usepackage{float}
\usepackage{multirow}
\usepackage{pstricks} 
\usepackage{color}
\usepackage{bbold}
\usepackage{xr} 
\usepackage{tikz}
\usepackage{comment}
\usetikzlibrary{matrix}
\usepackage{indentfirst}
\usepackage{siunitx}
\usepackage{systeme}
\usepackage{fancyhdr}
\usepackage{gensymb}
\usepackage[caption=false]{subfig}
\usepackage{mwe}
\usepackage{framed}
\usepackage{physics}
\usepackage{empheq}
\usepackage{url}
\usepackage{natbib}
\usepackage{lipsum}
\usepackage{stackengine}
\newcommand\xrowht[2][0]{\addstackgap[.5\dimexpr#2\relax]{\vphantom{#1}}}

\newtheorem{Definition}{Definition}
\newtheorem{Lemma}{Lemma}
\newcommand{\nn}{\nonumber} 
\newtheorem{Theorem}{Theorem}
\newcommand{\hmin}{H_{\mathrm{min}}}
\newcommand{\hmax}{H_{\mathrm{max}}}
\newcommand{\eqn}[1]{\begin{eqnarray} \newline #1 \end{eqnarray}}
\newcommand{\EV}[1]{\left < #1 \right >}
\newcommand{\adag}{\hat{a}^{\dag}}
\newcommand{\anih}{\hat{a}}
\newcommand{\hs}{\hspace{0.2cm}}
\newcommand{\half}{\frac{1}{2}}
\newcommand{\ee}{&=&}
\newcommand{\bk}[1]{\left ( #1\right )}
\definecolor{nathan}{rgb}{0,.8,.5}
\newcommand{\nw}[1]{{\color{black} #1}}

\begin{document}

\title{\nw{Certified Quantum Random Numbers from Untrusted Light}}

\author{David Drahi}
\email{david.drahi@physics.ox.ac.uk}
\affiliation{Clarendon Laboratory, Department of Physics, University of Oxford, Oxford OX1 3PU, UK}

\author{Nathan Walk}
\affiliation{Department of Computer Science, University of Oxford, Oxford OX1 3QD, UK}
\affiliation{Dahlem Center for Complex Quantum Systems, Freie Universit{\"a}t Berlin, 14195 Berlin, Germany}

\author{Matty J. Hoban}
\affiliation{Department of Computing, Goldsmiths, University of London, London SE14 6NW, UK}

\author{\nw{Aleksey K. Fedorov}}
\affiliation{Russian Quantum Center, 100 Novaya St., Skolkovo, Moscow 143025, Russia}

\author{\nw{Roman Shakhovoy}}
\affiliation{Russian Quantum Center, 100 Novaya St., Skolkovo, Moscow 143025, Russia}

\author{\nw{Akky Feimov}}
\affiliation{Russian Quantum Center, 100 Novaya St., Skolkovo, Moscow 143025, Russia}

\author{\nw{Yury Kurochkin}}
\affiliation{Russian Quantum Center, 100 Novaya St., Skolkovo, Moscow 143025, Russia}

\author{W. Steven Kolthammer}
\affiliation{Clarendon Laboratory, Department of Physics, University of Oxford, Oxford OX1 3PU, UK}

\author{Joshua Nunn}
\affiliation{Clarendon Laboratory, Department of Physics, University of Oxford, Oxford OX1 3PU, UK}

\author{Jonathan Barrett}
\affiliation{Department of Computer Science, University of Oxford, Oxford OX1 3QD, UK}

\author{Ian A. Walmsley}
\affiliation{Clarendon Laboratory, Department of Physics, University of Oxford, Oxford OX1 3PU, UK}

\date{June 3, 2020}

\begin{abstract}
{A remarkable aspect of quantum theory is that certain measurement outcomes are entirely unpredictable to all possible observers. Such quantum events can be harnessed to generate numbers whose randomness is asserted based upon the underlying physical processes. \nw{We formally introduce, design and experimentally demonstrate an ultrafast optical quantum random number generator that uses a totally untrusted photonic source}. While considering completely general quantum attacks, we certify and generate in real-time random numbers at a rate of \nw{$8.05\,$Gb/s} with a rigorous security parameter of \nw{$10^{-10}$}. Our security proof is entirely composable, thereby allowing the generated randomness to be utilised for arbitrary applications in cryptography and beyond. \nw{To our knowledge, this represents the fastest composably secure source of quantum random numbers ever reported.}}
\end{abstract}

\maketitle

\section{Introduction}
The inherent randomness of quantum theory, embodied by Born's rule, creates fundamentally unpredictable events. The concept of a quantum random number generator (QRNG) is to leverage this principle to produce a random, unpredictable output with an unparalleled level of confidence. The central challenge faced by practical QRNGs is to rigorously quantify how much of the entropy generated by a real-world device is indeed intrinsically unpredictable.

To sketch the basic idea, let's consider a device completely described by parameters $s$ which could be quantum or classical. These are used to generate a classical outcome $X$ that should appear unpredictable from the perspective of an agent external to the device. Consider such an agent $E$ with access to a system which includes all the parameters $s$ as well as any other side information (classical or quantum). For any given value of $s$, the joint system is described by a classical-quantum state $\hat{\rho}_{XE}$ and the outcome's predictability is simply the probability of the best guess
\eqn{P_{\mathrm{ideal},s}(X|E) = \sup_{\{\hat{E}_x\}} \sum_x p_x \mathrm{tr}\bk{\hat{E}_x\hat{\rho}_E^x}\label{eq1}\,,}
where the supremum is taken over all measurements $\{\hat{E}_x\}$ made by $E$ on the system and $\hat{\rho}_E^x$ is the state of $E$ conditioned on $X=x$. For a real device, however, $s$ is never known exactly. In this case, a conservative estimate of the predictability is given by $P=\max_s P_{\textrm{ideal},s}(X|E)$, where the maximisation is taken over all plausible parameters $s$. Confidence in the randomness is thus linked to claims about trusted workings of the device and subsequent constraints on the knowledge of the external agent. 

Approaches to QRNGs differ by the detail with which the devices need to be characterised in order to constrain $s$ \cite{ herrero2017quantum,ma2016quantum}. Perhaps the simplest conceptually is a so-called device independent QRNG, which can take the form of a Bell test \cite{pironio2010random, acin2016certified, bierhorst2018experimentally, liu2018device}. In this case, the device must be composed of two isolated measurements that employ independently selected bases --- a requirement that can be verified with high confidence. With this condition, $P<1$ as long as the measurement outcomes violate a Bell inequality, which in turn constrain the plausible $s$ \cite{acin2012randomness}. In reality, however, even state-of-the-art implementations \cite{liu2018high} are extremely complex and yield impractical bit rates of the order $\sim10\,$b/s. An alternate approach is to build a QRNG in which the entire device, from quantum source to measurement, is faithfully characterised and modelled \cite{mitchell2015strong}. Here, the detailed characterisation, which might use both off-line and in-line measurements, crucially constrains $s$ (and thus $E$) sufficiently to assert a non-unit $P$. As such, this seemingly exhaustive type of characterisation of the setup, and hence trust in its proper inner workings, opens up a myriad of potential attacks and malfunctions which might compromise the randomness output. 

A series of intermediate approaches have appeared, commonly referred to as having partial device-independence, which yield a QRNG that permits abstraction from some of the devices while needing a detailed characterisation of the remainder. These can be broadly classified as those that are independent of the measurement devices \cite{cao2015loss, chaturvedi2015measurement, nie2016experimental} or the sources \cite{cao2016source}. A third class, known as semi-device-independent makes no assumptions on either the source or measurements except to assert a global constraint on the relevant dimension \cite{pawlowski2011semi,lunghi2015self}, energy \cite{Himbeeck2017semidevice} or orthogonality of the relevant states \cite{brask2017megahertz}. Finally, other works have combined assumptions, such as the semi-source independent protocols (originally thought to be fully source-independent) that invoke a dimension assumption in conjunction with a calibrated detection \cite{vallone2014quantum,marangon2017source,michel2019real}. These latter works exemplify the critical point that when analysing partially device-independent protocols, it is important to keep track of the interaction between trusted, but imperfect, devices and the certification techniques used to prove security against deviations in the untrusted components. 

Successful design of a practical QRNG must balance confidence with ease of implementation, achievable bit rate, durability and cost. For example, QRNGs based on radioactive decay have limited bit rates, whereas those utilising electronic noise require careful distinction of quantum and thermal fluctuations \cite{herrero2017quantum}. In contrast, optical QRNGs promise well isolated quantum systems along with speed and technical ease. Implementations have been based on photon \textit{welcher weg} \cite{rarity1994quantum, jennewein2000fast, stefanov2000optical}, photon arrival time \cite{wayne2009photon, nie2014practical}, photon number statistics \cite{ren2011quantum}, vacuum fluctuations \cite{gabriel2010generator, shen2010practical, symul2011real}, phase noise \cite{guo2010truly, abellan2014ultra, nie2015generation} and Raman scattering \cite{bustard2013quantum, england2014efficient}.

In this paper, we develop a certification of quantum randomness generated by an optical beam splitter for which one input field is the vacuum and the other is completely unknown. The certification was carried out in real-time using an additional vacuum mode to tap off part of the unknown light source prior to the randomness generation. This method probabilistically infers a lower bound on the photon number of the remaining untrusted source impinging onto the randomness generation measurement. We show that signals from carefully characterised photodetectors, which needn't resolve photon number, are sufficient to both generate and certify genuine quantum randomness. 

Our approach results in a composably secure protocol and we provide an explicit security proof for high-speed quantum randomness expansion. Such a proof is necessary for all applications that wish to claim provable quantum-based security. A key or random string only becomes useful in composition with other protocols (one-time pad, hashing etc.) such that in order to retain provable quantum security, a composable proof is mandatory. To date, most randomness generation protocols fail to provide outputs that are useable in a composable framework, with, to our knowledge, only a handful shown to be composably secure in a device-dependent scenario \cite{mitchell2015strong,Haw:2015kx,Gehring:2018wc} and only one partially device independent result \cite{cao2016source}. 

\nw{To experimentally demonstrate our scheme, we used off-the-shelf components --- a laser source, high bandwidth photodiodes, basic linear optical elements and a high-performance field-programmable gate array (FPGA) board --- and generated random numbers with a bit rate of $8.05\,$Gb/s and a composable security parameter $\epsilon=10^{-10}$.} Overall, our framework is compatible with a wide range of optical detectors and avoids the need to trust or precisely characterise the source of light, as opposed to conventional vacuum homodyning wherein a trusted photonic source is a necessity. 

\section{Generating randomness from untrusted light}
In Eq.~(\ref{eq1}), we quantified the randomness of an outcome $X$ for an external agent $E$. As is common in quantum cryptography, we will refer to this agent as Eve the eavesdropper. An equivalent, but more convenient, way of quantifying this randomness is to compute the quantum conditional min-entropy of the quantum state $\hat{\rho}_{XE}$ for the joint system $XE$ \cite{renner2008security}
\eqn{\hmin(X|E)_{\hat{\rho}_{XE}} = -\log_2\bk{\sup_{\{\hat{E}_x\}} \sum_x p_x \mathrm{tr}\bk{\hat{E}_x\hat{\rho}_E^x}}\label{eq2}\,,}
where the argument of the logarithm is the guessing probability for Eve to guess $X$, as in Eq.~(\ref{eq1}). This quantity has been shown to quantify the number of bits --- almost perfectly random with respect to Eve --- that can be \textit{extracted} via post-processing \cite{konig2009operational}. Notice the distinction between a quantum randomness generator (QRG) which simply generates outputs with a certain conditional min-entropy and a QRNG that also includes the post-processing (hashing) necessary to produce almost perfect random numbers. \nw{This is worth mentioning because many results in the literature only implement the randomness generation without carrying out random number extraction in real-time. Note also that only by composably certifying the randomness generation process can the security of the extracted numbers be rigorously established.}

A certified randomness generation protocol allows for some, or all, devices to deviate arbitrarily from their purported specifications. \nw{A certification test $\mathcal{P}$ is applied to the experimental data and only upon that test passing is the output certified as having a certain amount of randomness, otherwise it is discarded.} Furthermore, a useful generator will be robust, i.e. it will pass the test with high probability. Formally, we can define such a protocol as follows.

\begin{Definition} \label{QRGdef}
\nw{An ($m,\kappa,\epsilon_{\mathrm{fail,m}},\epsilon_c$)-certified randomness generation protocol produces an output $X$ made of $m$ measurement results such that
\begin{itemize}
\item \textbf{Security:} Either the certification test $\mathcal{P}$ fails, or
\eqn{\hmin(X|E) \geq \kappa \nn\,,}
except with probability $\epsilon_{\mathrm{fail,m}}$.
\item \textbf{Completeness:} There exists an honest implementation such that the test will be passed with probability $1-\epsilon_c$.
\end{itemize}
}
\end{Definition}

\begin{figure}[h]
\includegraphics[width=\linewidth]{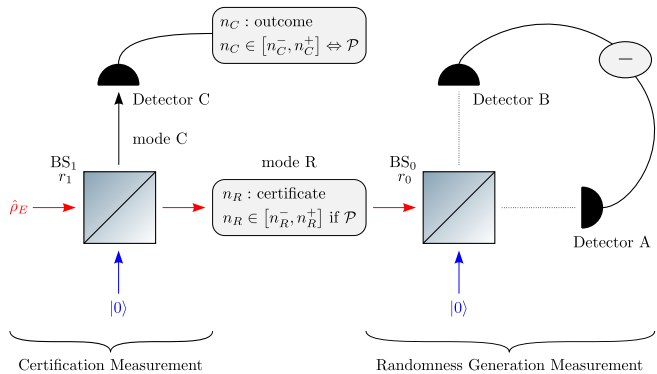}
\caption{Scheme for our SDI protocol. An unknown light source $\hat{\rho}_{E}$ is mixed with a trusted vacuum on a beam splitter (BS) with reflectivity $r_{1}$ to perform a certification measurement. The measured outcome at detector C is subject to a test $\mathcal{P}$ that passes if the outcome lies within a certain range $\left[n_{C}^{-},n_{C}^{+}\right]$. Upon passing the test, we certify a photon number $n_{R}$ in mode R that impinges onto the randomness generation measurement except with probability $\epsilon_{\mathrm{fail}}$.}
\label{fig:idea}
\end{figure}

We define our source-device-independent (SDI) photonic QRG as a protocol in which detectors and passive optical devices (e.g. beam splitters) are taken to be trusted. Photonic states are generated via a laser as input to the experiment (essentially preparing a large amplitude coherent state), however in the analysis, we will not assume anything about the state of these photons and in that sense we claim that randomness is generated in a SDI manner. Crucially, however, we also assume that it is possible to exploit a trusted vacuum mode. One might point out that this is in fact assuming at least one trusted source, namely the vacuum. Nevertheless, we argue that vacuum is a rather privileged source in the sense that it does not really require a ``device'' to be generated, merely the ability to block an input port to a beam splitter. Thus, it would seem highly preferable from a security perspective to trust a vacuum source rather than some photonic state created by a sophisticated device such as a laser or spontaneous parametric down conversion (SPDC) process. We also emphasise that the detection process here is distinct from a homodyne detection in that the incoming state is mixed with a vacuum mode instead of a local oscillator (large amplitude coherent state). Even more importantly, we model our measurements directly as opposed to the homodyne protocols \cite{vallone2014quantum,marangon2017source,michel2019real} which model this detection as a quadrature measurement. \nw{This is rather at odds with the goal of being SDI as that is only approximately true in the limit where one assumes that the input signal has far fewer photons than the local oscillator}. In Section \ref{discussion} and Appendix \ref{comp}, we will also discuss how our measurement scheme has different, and in many cases, superior scalings of the certifiable randomness rates than standard homodyne based protocols.

To gain some intuition, let us start by considering the randomness generation measurement depicted in Fig.~\ref{fig:idea}. It consists of a beam splitter BS$_{0}$ with reflectivity $r_{0}=\frac{1}{2}$, an input mode R, a trusted vacuum fed into the other input mode and two output photodetectors A and B performing a difference measurement. Assuming the photodetectors to be perfect, we can model them as performing a single measurement acting on the untrusted photonic randomness source in mode R. The outcomes of the measurement will be the photon numbers $n_{A}$ and $n_{B}$ detected by detectors A and B, respectively. Propagating this detection event back through the beam splitter and using our knowledge about the trusted vacuum mode, this measurement is then associated with positive-operator valued measure (POVM) elements of the form
\begin{equation}
\hat{M}(n_{A},n_{B})_{R}=\frac{(n_{A}+n_{B})!}{2^{n_{A}+n_{B}}n_{A}!n_{B}!}\ket{n_{A}+n_{B}}\bra{n_{A}+n_{B}}_{R}\,,
\label{eq3}
\end{equation}
living in the Hilbert space of the input mode R (see Appendix \ref{ideal} for details). 

Given this, we now propose a simple certifiable randomness generation protocol. It consists of recording the value of the photon number sum $N:=n_{A}+n_{B}$ and then using the difference measurement $x:=n_{A}-n_{B}$ as the source of randomness. Therefore, we have two measurements: one of $N$ and one of $x$. The POVM $\mathbb{Z}$ has elements $\hat{Z}(N)$ for the measurement of $N$ that can be readily recovered as 
\begin{equation}
\begin{split}
\hat{Z}(N)=&\sum_{n_{A}=0}^{N}\hat{M}(n_{A},N-n_{B})_{R}\\
=&\ket{N}\bra{N}_{R}\,.
\end{split}
\end{equation}

On the other hand, as we show in Appendix \ref{ideal}, the POVM $\mathbb{X}$ for the value of $x$ has elements given by
\eqn{\hat{X}(x) \ee \sum_{n_{A}=|x|}^{\infty} 2^{-(2n_{A}-|x|)}{2n_{A}-|x|\choose n_{A}} \nn \\
&\times& \ket{2n_{A}-|x|}\bra{2n_{A}-|x|}_{R} \,. \label{xmain}}

We already see the inherent randomness of this scheme since $\hat{X}(x)$ has support over the whole Fock space. Therefore, for any state in mode R with total photon number $N>0$, there will be multiple possible values $x$ which can occur. Moreover, there is a manifest independence from the photonic input state. Because the measurements described by $\hat{Z}(N)$ and $\hat{X}(x)$ are by definition compatible, we can always think of the $\hat{Z}(N)$ measurement happening first and projecting onto the state $\ket{N}$, which will subsequently produce randomness when measured with $\mathbb{X}$. Thus, conditioned upon observing a sum value of $N$, one would certify with probability $\epsilon_{\mathrm{fail,m}} = 0$ an amount of randomness that scales as $\log_2(N\pi/2)$ as per Definition \ref{QRGdef} and shown in Appendix \ref{ideal}.

Now, consider the full setup shown in Fig.~\ref{fig:idea}. We introduce the certification measurement in mode C which is done by tapping off a fraction of the completely unknown incoming light in mode E with a beam splitter BS$_{1}$ of reflectivity $r_{1}$. The input state $\hat{\rho}_{E}$ is mixed with a trusted vacuum on BS$_{1}$ and the reflected beam in mode C is measured at detector C while the transmitted beam in mode R is input to the randomness generation measurement. \nw{This idea is superficially similar to the ``energy test'' proposed  in the context of device-dependent continuous variable quantum key distribution (QKD) \cite{furrer2014reverse}. This test also taps off a portion of the incoming mode but instead uses a trusted and ideal heterodyne detection for the certification measurement. Such a scheme is {\it a priori} forbidden in an SDI context (a trusted photonic source being necessary for a heterodyne detection) and, as we show in Appendix \ref{detectors}, also fails to provide any security for realistic finite-range detectors.}

Our test $\mathcal{P}$ is applied to the output of detector C with the protocol aborting if the result lies outside a range $[n_{C}^-,n_{C}^+]$. Upon passing the test, we obtain a certificate that $n_R$, the photon number in mode R, lies within a range $[n_{R}^-,n_R^+]$ except with some failure probability $\epsilon_{\mathrm{fail}}$. Then, by minimising the min-entropy over all states within this range, we obtain a certified lower bound on the generated randomness. For this idealised scenario, we could allow $n_R^+$ to be unbounded and would simply look to certify the largest possible value of $n_R^-$ given a specific $\epsilon_{\mathrm{fail}}$.

\section{Certifying randomness with realistic devices}
In a real experiment, several further complications must be taken into account. Even in a scenario of completely trusted and calibrated devices, care must be taken to quantify the amount of randomness that can be credibly claimed to have been generated. Firstly, real detectors only possess a finite dynamic range over which their response is meaningful. Secondly, measurement outcomes are coarse grained to a finite resolution which must be carefully accounted for when determining the output randomness. Finally, noisy devices will exhibit fluctuations due to processes not under complete experimental control. Information about these processes might be accessible to external observers and, even if not, could certainly be stemming from physical processes that are far from random. Nevertheless, this can be accounted for provided the device noise is calibrated and not controlled by Eve. This makes the noise essentially classical, in the sense that we may assume that it is described by variables $\lambda$ which are distributed according to a characterised probability distribution. These variables are then given to Eve on a shot-by-shot basis.   

Consequently, the first step for analysing our experiment is to carefully calibrate and model the realistic photodiodes, which output noisy voltage measurements rather than exact photon numbers. More formally, following the approach of \cite{frauchiger2013true}, we model the POVM describing our noisy, characterised measurements as a projective measurement on a larger system. For the case of our detectors (see Fig.~\ref{fig:detector_model} in Appendix \ref{detectors} for a cohesive summary), the measured voltages are modelled as follows. First, we consider an $L := n_{\mathrm{max}} - n_{\mathrm{min}} + 1$ outcome photon number resolving measurement with a finite range $[n_{\mathrm{min}},n_{\mathrm{max}}]$ described by measurement operators that are number state projectors (i.e. $\hat{N}(n) = \ket{n}\bra{n} $), except for the first and last operators which are given by $\hat{N}(n_{\mathrm{min}}) = \sum_{n=0}^{n_{\mathrm{min}}}\ket{n}\bra{n}$ and $\hat{N}(n_{\mathrm{max}}) = \sum_{n=n_{\mathrm{max}}}^\infty \ket{n}\bra{n}$. \nw{This photon number is converted to a voltage via a conversion factor $\alpha$ and is then smeared by an additional Gaussian noise term $\lambda$ of known variance $\sigma^2$ and finally coarse grained by a $b$-bit analogue-to-digital converter (ADC) that itself has only finite range $[V_{\mathrm{min}},V_{\mathrm{max}}]$ and finite resolution of $2^{b}$ bins. However, to correctly quantify the randomness associated with each $b$-bit measurement, it is essential for one to consider $\Delta_{\mathrm{ADC}}$, the ADC's effective number of bits (ENOB). Indeed, it corresponds to the amount of bits free of internal electronic noise. This effective bit depth leads to an effective voltage resolution $\delta V = \frac{V_{\mathrm{max}}-V_{\mathrm{min}}}{2^{\Delta_{\mathrm{ADC}}}}$}. The output of such a realistic measurement is an index, say $j$, corresponding to a voltage bin of width $\delta V$ centered at $j\delta V$. We can therefore associate minimum and maximum voltages $v_j^\pm = \delta V(j\pm \frac{1}{2})$ with this outcome $j$. 

The certification measurement is made by mixing the unknown photonic input $\hat{\rho}_{E}$ in mode E with vacuum $\ket{0}$ on a beam splitter of reflectivity $r_{1}$. The reflected mode C is then detected with a noisy photodiode (characterised by noise standard deviation $\sigma_C$ and voltage conversion factor $\alpha_C$) that is coarse grained by an ADC. The protocol aborts for sufficiently large or small observed voltages ($\mathcal{P}$ is now a test applied directly to the measured voltage index). Finally, the randomness is generated by mixing the transmitted state in mode R with another vacuum on a beam splitter with reflectivity $r_{0}=\frac{1}{2}$ and making a coarse-grained, noisy difference measurement characterised by noise standard deviation $\sigma_D$ and voltage conversion factor $\alpha_D$. As with the ideal case, we can write the measurements as operators in the input Hilbert space. As shown in Appendix \ref{detectors}, the POVM element for a realistic voltage difference measurement whose outcome is the bin labelled $j$ is 

\eqn{\hat{V}_D^{\sigma_D,\Delta_{\mathrm{ADC}}}(j) =  \int_{I_j^D} \hat{V}_D^{\sigma_D}(v_{D}) \, dv_{D} \label{measfinal_main} \,,}
with 
\eqn{\hat{V}_D^{\sigma_D}(v_D) = \sum_{x=-(L-1)}^{L-1} \frac{e^{-(v_{D}-\alpha_{D}x)^2/(2\sigma_D^2)}}{\sqrt{2\pi}\sigma_D} \hat{X}_{\mathrm{fin}}(x)\label{volts_main}\,,}
where $\hat{X}_{\mathrm{fin}}(x)$ are the POVM elements of a difference measurement that is identical to Eq.~(\ref{xmain}) except that it is made with finite range photodetectors described above and is hence only operationally equivalent over an input photon number range $[n_{\mathrm{min}}^D,n_{\mathrm{max}}^D]$.

Similarly, the certification measurement element corresponding to the outcome bin labelled $i$ is given by
\eqn{\hat{V}_C^{\sigma_C,\Delta_{\mathrm{ADC}}}(i) = \int_{I_i^C} \hat{V}_C^{\sigma_C}(v_C) \, dv_C\label{certfinal_main} \,,}
with
\eqn{\hat{V}_{C}^{\sigma_C}(v_C) = \sum_{n=n^C_{\mathrm{min}}}^{n^C_{\mathrm{max}}} \frac{e^{-(v_C-\alpha_C n_C)^2/(2
\sigma_C^2)}}{\sqrt{2\pi}\sigma_C} \hat{N}_C(n_C)\label{certmain}\,.}

\nw{With this detection model in hand, we state our main theorem as follows.}
\begin{Theorem} \label{rgend1} 
An optical setup consisting of 
\begin{itemize}
\item Two trusted vacuum modes
\item Two beam splitters of reflectivity $r_{0}=\frac{1}{2}$ and $r_{1}$
\item Two noisy photodetectors used to make a difference measurement as described in Eq.~(\ref{measfinal_main})
\item A third noisy photodetector used to make a certification measurement as described in Eq.~(\ref{certfinal_main}) which passes the test $\mathcal{P}$ if $i$ falls in a chosen range $[i_-,i_+]$
\end{itemize}

can be used as a certified (m,$\kappa$,$\epsilon_{\mathrm{fail,m}}$,$\epsilon_c$)-randomness generation protocol as per Definition \ref{QRGdef} without making any assumptions about the photonic source with
\eqn{\kappa &\geq& - m\log_2 \left ( \sum_{x\in \mathcal{X}}2^{-n_{R}^{-}}\binom{n_{R}^{-}}{\lfloor \frac{n_{R}^{-}+x}{2} \rfloor} \right )\label{hminthmmain}\,,}
where
\eqn{\mathcal{X} \in \mathbb{N} \cap \left [ -\left \lfloor \frac{\delta V}{2\alpha_{D}} \right \rfloor, \left \lfloor \frac{\delta V}{2 \alpha_{D}} \right \rfloor  \right ] \label{supportmain}\,,}
with $\delta V = \frac{V_{\mathrm{max}}-V_{\mathrm{min}}}{2^{\Delta_{\mathrm{ADC}}}}$,

\eqn{\label{mroundmain} \epsilon_{\mathrm{fail,m}} &\leq&m \epsilon_{\mathrm{fail}} \,,}
where
\eqn{\epsilon_{\mathrm{fail}} = \max \{ \epsilon_-,\epsilon_+\} + \epsilon_{\lambda_{C}} \label{epsilontheoremmain}\,,}
with 
\eqn{ \epsilon_- \ee \exp\bk{-2\frac{\left(\frac{v_{i_{-}}^- - \tilde{\lambda}}{\alpha_C}-r_{1} \left(\frac{v_{i_{-}}^- - \tilde{\lambda}}{\alpha_C} + n_{R}^{-} - 1\right)\right)^2}{\frac{v_{i_{-}}^- - \tilde{\lambda}}{\alpha_C} + n_{R}^{-} - 1}} \,, \nn \\
\epsilon_+ \ee \exp\bk{-2\frac{\left(n_{R}^{+} -(1-r_{1}) \left(\frac{v_{i_{+}}^+ - \tilde{\lambda}}{\alpha_C} + n_{R}^{+} + 1\right)\right)^2}{\frac{v_{i_{+}}^+ - \tilde{\lambda}}{\alpha_C} + n_{R}^{+} + 1}} \,, \nn\\ 
\epsilon_{\lambda_{C}} \ee 1- \mathrm{erf}\left(\frac{\tilde{\lambda}}{\sqrt{2} \sigma_C }\right) \label{esecrealmain} \,,}
provided $n_R^+$ is set to the saturating photon number of the difference measurement. 

Moreover, \eqn{\epsilon_c = 1 - \mathrm{tr} \left\{\sum_{i=i_-}^{i_+}\ket{\alpha}\bra{\alpha} \hat{V}_C^{\sigma_C,\Delta_{\mathrm{ADC}}}(i) \right\} \label{completemain} \,,} 
using a coherent state $\ket{\alpha}$ as an input.
\end{Theorem}

\noindent{\it Proof sketch:} 
For a complete proof, see Appendix \ref{proof_main_theorem}. \nw{The protocol consists of $m$ rounds, each of which are defined as a certification measurement subjected to the test $\mathcal{P}$ and a randomness measurement sample that is registered in $X$.} One part of the proof is to show that, for any given round of the protocol, conditioned on passing the test $\mathcal{P}$, the state in mode R has support in the photon number basis that lies almost entirely in the range $[n_R^-,n_R^+]$. More concretely, we maximise over all possible input states to upper bound
\eqn{\epsilon_{\mathrm{fail}} := \max_{\hat{\rho}_E} \mathrm{Pr} \left[i^{-}\leq i \leq i^{+} \wedge n_R\notin [n_{R}^{-},n_{R}^{+}]\right]\,,}
the joint probability that the test would be passed in mode C whilst a photon number outside the range $[n_{R}^-,n_R^+]$ was present in mode R. This quantity can be interpreted as the probability that the conditional state in mode R can be operationally distinguished from any state solely supported within $[n_{R}^{-},n_{R}^{+}]$ (see Appendix \ref{details}). 

The second part of the proof is to optimise over all possible input states with support only in $[n_{R}^{-},n_{R}^{+}]$ to derive a lower bound on the conditional min-entropy. Note that {\it a priori}, Eve has the freedom to choose an input state that is potentially entangled across all $m$ rounds, i.e. we are considering completely general, so-called coherent attacks. Together, these results mean that either the min-entropy for a single round will be lower bounded or the protocol will abort except with probability $\epsilon_{\mathrm{fail}}$. For $m$ rounds, one can simply add these lower bounds together to bound the min-entropy of the output \nw{concatenated} string except with a probability 
\eqn{\epsilon_{\mathrm{fail, m}} := 1-(1-\epsilon_{\mathrm{fail}})^m\leq m\epsilon_{\mathrm{fail}} \,, \label{proof1}} 
as claimed in Eq.~(\ref{mroundmain}).

Intuitively, one would expect that Eve's optimal strategy to predict the outcome of a difference measurement would be to input a pure Fock state and this is indeed the case. The key fact is that the realistic difference measurement is still diagonal in the photon number basis and that a $m$-round protocol can be described as a tensor product of such measurements. Note that for the purposes of calculating the min-entropy, we consider the difference measurement in Eq.~(\ref{measfinal_main}) from the perspective of Eve who knows the noise variable $\lambda_D$ on a shot-by-shot basis, for which $\hat{V}_D^{\Delta_{\mathrm{ADC}}}(j) = \sum_{x\in \mathcal{X}} \hat{X}(x)$, where $\mathcal{X} = \{x : \alpha_{D} x + \lambda_{D} \in I_{j}^D \}$.
The fact that this measurement commutes with a diagonalising map in the photon number basis makes it straightforward to show that Eve's optimal guessing probability is achieved by inputting a pure Fock state. Provided we choose $n_R^+$ less than $n_{\mathrm{max}}$, the saturation value for the detectors, then direct calculation shows that the guessing probability decreases monotonically in $n_R$. Thus, for states restricted to $[n_{R}^{-},n_{R}^{+}]$, the smallest min-entropy is achieved by inputting $\ket{n_R^-}$. Finally, the fact that the coefficients in Eq.~(\ref{xmain}) are those of a binomial distribution can be used to show that Eve's min-entropy is minimised whenever $x$ is minimal (0 or 1 depending if an odd or even photon number is input) and $\lambda_D = 0$. Assuming that this is always the case, direct evaluation of $\mathrm{tr} \left \{ \ket{n_R^-}\bra{n_R^-} \hat{V}_D^{\Delta_{\mathrm{ADC}}}\,(n_R^-\mod 2)\right \}$ yields the expression in Eq.~(\ref{hminthmmain}).

Turning to the failure probability, we first define a failure operator which corresponds to taking the failure condition (i.e. a passing voltage is observed at detector C along with $n_R\notin [n_{R}^{-},n_{R}^{+}]$ in mode R) and write it as an operator in the Hilbert space of Eve's input mode 
\eqn{\hat{V}_F^{\Delta_{\mathrm{ADC}}}(i,n_{R}^{-},n_{R}^{+}) \ee \sum_{\stackrel{n_C\in \mathcal{C}}{n_R\notin [n_{R}^{-},n_{R}^{+}] }} \frac{r_{1}^{n_C}(1-r_{1})^{n_R}(n_C+n_R)!}{n_C!n_R!} \nn \\
&\times& \ket{n_C+n_R}\bra{n_C+n_R}_E \,. \label{fail}}

Since this operator is also diagonal in the photon number basis, one can repeat the previous arguments to show that Eve's optimal strategy to maximise this failure probability is also achieved by a Fock state.

The failure probability for a single round of the protocol can then be written as 
\eqn{\epsilon_{\mathrm{fail}}\ee \max_{n_E}  \sum_{i= i^{-}}^{i^{+}} \bra{n_E}\hat{V}_F^{\sigma_C,\Delta_{\mathrm{ADC}}}(i,n_{R}^{-},n_{R}^{+})\ket{n_E} \,, \label{failsketch}}
where $\mathcal{C} = \{ n_C: \alpha_C n_C + \lambda_C \in [i^-,i^+] \}$. 

To bound this quantity, we first use our knowledge of the certification noise variable $\lambda_C$. Except with probability $\epsilon_{\lambda_{C}} = 1 - \mathrm{erf}\left(\frac{\tilde{\lambda}}{\sqrt{2} \sigma_C }\right)$, we know that $|\lambda_C|\leq \tilde{\lambda}$. Substituting Eq.~(\ref{fail}) in Eq.~(\ref{failsketch}) yields two terms as the sum over $n_R\notin [n_R^-,n_R^+]$ decomposes as a sum for $0\leq n_R<n_R^-$ and $n_R^+<n_R\leq \infty$. Provided we have $\lambda_C \leq v^+_{i^+} - \alpha_C\bk{n_R^+ - n_R^- + 1}$, then there is no value of $n_E$ for which both terms will be simultaneously non-zero and we can write 
\eqn{\epsilon_{\mathrm{fail}} = \max \{ \epsilon_-,\epsilon_+\} + \epsilon_{\lambda_{C}} \,, \label{proof2}} 
where $\epsilon_-$ ($\epsilon_+$) corresponds to the lower (upper) sum. 

Both of these are essentially cumulative binomial distributions. For example, for a particular value of $n_E$
\eqn{\epsilon_- \leq \sum_{\stackrel{n_C=\max \{ n_C^-,}{n_E-(n_{R}^{-}-1)\}}}^{n_E} \frac{r_{1}^{n_C}(1-r_{1})^{n_R}(n_C+n_R)!}{n_C!n_R!} \,, \label{eminus}}
where $n_C^-$ is the smallest photon number allowed at mode C consistent with passing the test. 

For unbounded $\lambda_C$, it would be impossible to determine $n_C^-$ or $\epsilon_-$, but again using $\tilde{\lambda}$, we can do so except with probability $\epsilon_{\lambda_C}$. If we define $v_i^{-(+)}$ as the minimum (maximum) voltage compatible with the passing range $[i^-,i^+]$, we can obtain a minimum (maximum) photon number $n_C^- = (v_i^--\tilde{\lambda})/\alpha_C$ ($n_C^+ = (v_i^++\tilde{\lambda})/\alpha_C$) for mode C compatible with passing the test. The varying lower limit on the sum in Eq.~(\ref{eminus}) stems from the fact that for Eve to cheat, there are two constraints on $n_C$. First, it must be the case that a sufficiently large number of photons go to detector C such that the test is passed, but for sufficiently large $n_E$ this condition is superseded by the requirement that less than $n_R^-$ photons go to mode R. Arguments based upon the nature of the binomial coefficients allow us to show that to maximise $\epsilon_-$, Eve should choose the input state $n_E^{\mathrm{opt}} = n_C^- + n_R^- - 1$. This can be directly substituted into Eq.~(\ref{eminus}) and the application of Hoeffding's bound yields the term appearing in Eq.~(\ref{esecrealmain}). Finally, an analogous argument can be applied to bound $\epsilon_+$ as per Eq.~(\ref{esecrealmain}). In combination with Eq.~(\ref{proof1}) and Eq.~(\ref{proof2}), this completes the security proof.

\section{Extracting Random Numbers from Certified Quantum Randomness\label{sec:hash}}
\nw{Finally, we turn to the task of actually extracting $\epsilon$-secure random numbers for use in real-world applications. This can be achieved via two-universal hashing (detailed in Appendix \ref{rne}) which can be efficiently implemented using an FPGA. The details of the randomness extraction are critical in determining both the final speed and security of the QRNG. Firstly, one must obtain a composable certificate for how close the hashed outputs are to perfect randomness. Secondly, one needs to assess whether the randomness extraction is performed in real-time, i.e. at a rate greater than or equal to the randomness generation rate posed by the experiment. To precisely address these issues, the critical parameters are the FPGA's hashing speed (number of hashes per second) and the hashing block size. 

Regarding the composable security definition for the final hashed numbers, we can simply adopt the following standard secrecy criteria from the QKD literature \cite{portmann2014cryptographic}.

\begin{Definition} \label{QREdef}
Let $X$ be the random variable describing the measurements of a certified QRG protocol which succeeds with probability $p_{\mathrm{pass}}$ and let $S$ denote the result of a randomness extraction process applied to $X$. The result $S$ is $\epsilon$-secure if $\hat{\rho}_{SE}$, the joint state with the eavesdropper, satisfies
\eqn{p_{\mathrm{pass}} D(\hat{\rho}_{SE}, \hat{\rho}_{\mathrm{ideal}}) \leq \epsilon\,, \label{eq:QREdef}}
where $D(\hat{\rho},\hat{\sigma}):=\frac{1}{2}||\hat{\rho}-\hat{\sigma}||_1$ is the trace distance and $\hat{\rho}_{\mathrm{ideal}}$ is the output of an ideal randomness source, defined as $\hat{\rho}_{\mathrm{ideal}} := \hat{\tau}_S\otimes\hat{\rho}_E$, with $\hat{\tau}_S$ the uniformly distributed state on $S$. 
\end{Definition}

Due to the composable nature of our randomness generation protocol, we can apply previous results on hashing with quantum side information \cite{tomamichel2011leftover} to obtain the desired certificate in Eq.~(\ref{eq:QREdef}). Its precise formulation is given by the theorem below (see Appendix \ref{rne} for a full derivation).

\begin{Theorem}
A certified SDI (m,$\kappa$,$\epsilon_{\mathrm{fail,m}}$,$\epsilon_c$)-randomness generation protocol as defined in Definition \ref{QRGdef} can be processed with a random seed of length $m$ via two-universal hashing to produce a certified SDI random string of length $l$ given by
\eqn{l \ee \kappa + 2 -\log_2 \frac{1}{\epsilon_{\mathrm{hash}}^2}\label{lexp} \,,}
that is $\epsilon_c$-complete and $\epsilon_l = \epsilon_{\mathrm{hash}}+\epsilon_{\mathrm{fail,m}}$ secure
with
\eqn{ \epsilon_{\mathrm{hash}} = 2^{(l - \kappa - 2)/2}\,. \label{eq:hashthm}}
\end{Theorem}

To understand how such a system will perform, we will examine these security parameters in more detail beginning with $\epsilon_{\mathrm{hash}}$. The raw data output by an $m$-round QRG protocol will be a bit-string of length $h=mb$, where $b$ is the total number of bits recorded by the ADC for each measurement (recall that this is different from $\Delta_{\mathrm{ADC}}$, the effective number of noise free bits that we used to lower bound the randomness). From Theorem \ref{rgend1}, we know that the total min-entropy is proportional to the number of rounds, or alternatively the block length, and so we can write $\kappa = g'm = \frac{g'}{b} h:=gh$ for some constants $g$ and $g'$. The extracted length can also be written in terms of a compression ratio $r$ defined by $l = r\times h$. Putting this together, we can rewrite Eq.~(\ref{eq:hashthm}) as 
\eqn{\epsilon_{\mathrm{hash}} = {2^{-\frac{h}{2}(g-r-2)} \,.}\label{hashblock}}

To see the critical importance of the block size $h$, consider the case of maximal compression. For fixed $h$, there is hard lower limit to the compression ratio given by $r\geq\frac{1}{h}$, since the minimum possible output length is $1$ bit. This in turn necessitates a lower limit $\epsilon_{\mathrm{hash}} \geq 2^{-\frac{hg-3}{2}}$ and hence a limit on the total achievable $\epsilon_l$. This shows that a certain minimum block size is mandatory to obtain a given level of security. More generally, considering Eq.~(\ref{hashblock}), it becomes clear that increasing $h$ allows us to either increase the compression ratio while keeping $\epsilon_{\mathrm{hash}}$ constant (i.e. linearly improving performance whilst maintaining security) or decrease $\epsilon_{\mathrm{hash}}$ while keeping $r$ constant (i.e. exponentially improving security whilst maintaining performance). 

There is a further consideration in that augmenting the block size $h$ (i.e. taking more measurement samples $m$) has the deleterious effect of increasing the value of $\epsilon_{\mathrm{fail,m}}$. This can be compensated by either altering the voltage thresholds used in the test $\mathcal{P}$ at the cost of a decreased probability of passing the test $1-\epsilon_c$, or inferring a smaller certified minimum photon number and hence a smaller min-entropy $\kappa$. This in turn feeds back into $\epsilon_{\mathrm{hash}}$. Nevertheless, although one cannot arbitrarily increase $h$, in practice it turns out that having a sufficiently large block size is imperative for maximising the overall performance of a QRNG setup. If the min-entropy per measurement is relatively low, then as per Eq.~(\ref{hashblock}) and the discussion above, a small $h$ prohibits any randomness extraction whatsoever. As well as this in-principle limitation, in practice, the maximum achievable block size $h$ is typically limited by the technical parameters of the FPGA used for post-processing. 


Therefore, depending upon the desired application, one may need to concatenate several blocks of hashed random numbers to obtain a final string of the requisite length. Intuitively, it should be possible to deliver shorter strings at a faster bit rate, given that less concatenation is required and hence worse security per hashed output string of length $l$ can be tolerated. Defining $t$ to be the number of output $l$-bit concatenated blocks, one obtains a final string of the desired length $L=t\times l = t\times r\times h$ with an overall security parameter $\epsilon$ given by
\eqn{\epsilon = t\epsilon_l \geq t (\epsilon_{\mathrm{hash}} + m\epsilon_{\mathrm{fail}})\,, \label{eq:concatenate}}
as per Eq.~(\ref{proof1}) and Eq.~(\ref{ese2}) in Appendix \ref{rne}.

One can now readily observe that for a fixed final $\epsilon$, a smaller number of concatenations $t$ would allow a larger value for $\epsilon_{\mathrm{fail}}$ and $\epsilon_{\mathrm{hash}}$ which in turn permits a larger compression ratio $r$ and thus a faster overall bit rate.


Turning to the final bit rate, there are two cases, depending upon whether it is the FPGA or the experiment itself which is the bottleneck. Consider the case when the hashing speed is slower than the experiment's output data generation rate. Define $R_{\mathrm{hash}}$ as the FPGA clock rate (i.e. the inverse of the time it takes to carry out one hashing operation). Since each hashing operation outputs $l$ bits, the total bit rate is 
\eqn{R_h:=R_{\mathrm{hash}}\times l = R_{\mathrm{hash}}\times r\times h\,,}
where the subscript $h$ denotes that the limiting time factor is the \textit{hashing} speed. 

The second case, which will hold for our real-time implementation, is when the experiment is slower than the hashing. Given an experimental data acquisition rate of  $R_{\mathrm{data}}$, the total bit rate will simply be
\eqn{R_d := R_{\mathrm{data}} \times r\,,}
where the subscript $d$ denotes that this time, it is the \textit{data} acquisition rate which is the limiting factor.

Ultimately, given that an honest implementation of the QRNG protocol passes with probability $1-\epsilon_{c}$, the averaged generated bit rate is 
\eqn{\EV{R} = (1-\epsilon_{\mathrm{c}}) \times \min \{R_h,R_d \}\,,}
where the minimum discriminates between the two possible cases described above.
}

\section{Experiment}
\nw{The experiment carries out two separate key tasks: the randomness generation and the real-time extraction of random numbers. 

\begin{figure}[h!]
\includegraphics[width=\linewidth]{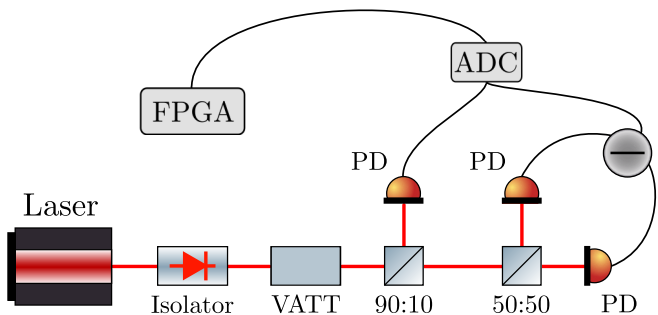}
\caption{\nw{Schematic of the setup used for random number generation. Measurements generated by the fibre-connected optical elements are fed to an ADC coupled to an FPGA. VATT: variable optical attenuator; PD: photodiode; ADC: analog-to-digital converter; FPGA: field-programmable gate array.}}
\label{fig:optical_setup}
\end{figure}

The experimental setup is displayed in Fig.~\ref{fig:optical_setup} and consists of a fully fibre-connected architecture with commercially available components for the randomness generation, and a high-speed field-programmable gate array (FPGA) for random number extraction. Note that for the randomness generation experiment, measurement signals will be analysed with an oscilloscope in order to precisely characterise the randomness found in each measurement while the real-time extraction of random numbers will be faithfully performed on a dedicated high-performance post-processing board containing both an ADC and an FPGA.}

\subsection{\nw{Randomness Generation}}
\label{sec:randomness_generation_experiment}
The light source utilised is a continuous wavelength (CW) laser (Koheras Adjustik E15) at telecom wavelength $\lambda=1550\,\si{\nano\metre}$. Note that the source's linewidth is less than $100\,\si{\hertz}$, thereby ensuring it to be effectively single-frequency. The laser output is directed onto a fibre optical isolator (Thorlabs IO-H-1550APC) in order to prevent unwanted back reflections into the laser. A fibre optical variable attenuator (model MAP-220CX-A from JDSU) is used to generate different photon numbers impinging onto the QRG by varying the laser's optical power. The certification and randomness generation measurements are implemented using standard fibre couplers (Thorlabs 10202A optimised for telecom wavelength) with reflectivities $r_{1}=0.0965$ (i.e. $\approx$ 90:10) and $r_{0}=\frac{1}{2}$ (i.e. 50:50), respectively. Detector C --- used for the certification measurement --- is a fibre-coupled InGaAs PIN photodiode (Thorlabs DET08CFC/M) with a large bandwidth $BW_{C}=5\,\si{\giga\hertz}$, a responsitivity $\eta_{C}=1.04\,\si{\ampere\per\watt}$ at $\lambda=1550\,\si{\nano\metre}$, a transimpedance gain $G_{C}=50\,\si{\ohm}$ and a measured electronic noise with standard deviation \nw{$\sigma_{C}\approx 0.25\,\si{\milli\volt}$}. On the other hand, the randomness generation measurement made of detectors A and B is implemented by means of a fibre-coupled AC-coupled balanced detector (Thorlabs PDB-480C-AC) with the following corresponding specifications: $BW_{D}=1.6\,\si{\giga\hertz}$, $\eta_{D}=0.95\,\si{\ampere\per\watt}$ at $\lambda=1550\,\si{\nano\metre}$, $G_{D}=16000\,\si{\ohm}$ and $\sigma_{D}\approx 3.05\,\si{\milli\volt}$. \nw{Signals from the detectors are sampled by an oscilloscope (Lecroy WaveRunner 204MXi) with a $2\,\si{\giga\hertz}$ bandwidth, a sampling rate of $F_{S}=10\,$GS/s and a voltage resolution of $V_{\mathrm{max}}-V_{\mathrm{min}}=10\,\si{\milli\volt}/\textnormal{div}$. The measurements are recorded by the oscilloscope's ADC as an 8-bit output, but with a calibrated bit depth of $\Delta_{\mathrm{ADC}}=4.772\,$bits.} This corresponds to the effective number of bits free of ADC internal noise. A total of 24 data sets were acquired, scanning the optical power input to the difference measurement from $0\,\si{\milli\watt}$ to $6.77\,\si{\milli\watt}$, corresponding to the balanced detector's linearity response range. \nw{Each data set was acquired over $T=1\,\si{\milli\second}$, yielding 10 million samples per power setting.} 

To evaluate the certified randomness of this data for a desired failure probability $\epsilon_{\mathrm{fail}}$, we must first fix $\tilde{\lambda}$ such that $\epsilon_{\lambda_C}<\epsilon_{\mathrm{fail}}$ (here we choose $\epsilon_{\lambda_C} = \epsilon_{\mathrm{fail}}/2$). Then, given the difference measurement's saturation power, we set $n_R^+$ equal to the corresponding saturating photon number $n_{\mathrm{max}}^D = 1.06\times 10^7$ and choose an upper voltage threshold $v_{i_+}$ in Eq.~(\ref{esecrealmain}) such that $\epsilon_+<\epsilon_{\mathrm{fail}}/2$. Finally, for a given lower voltage threshold $v_{i_-}$, we solve Eq.~(\ref{esecrealmain}) to find $n_R^-$ such that $\epsilon_- = \epsilon_{\mathrm{fail}}/2$. This ensures that the photon number input to the difference measurement lies within $[n_R^-,n_R^+]$ except with probability $\max\{ \epsilon_-,\epsilon_+\} + \epsilon_{\lambda_C} = \epsilon_- + \epsilon_{\lambda_C} =\epsilon_{\mathrm{fail}}$ and the certified randomness can then be determined by plugging $n_R^-$ into Eq.~(\ref{hminthmmain}) to retrieve the conditional min-entropy. 

This establishes the protocol's SDI security as per Definition \ref{QRGdef}. However, to understand how much randomness we can expect to obtain in practice, we should also consider the protocol's completeness. Typically, we will have some claimed specifications for the source and can choose thresholds accordingly. We would normally only attempt to certify a quantity and quality of randomness such that the corresponding test $\mathcal{P}$ would be passed with high probability by a source satisfying the claimed specifications using Eq.~(\ref{completemain}). Here, for simplicity, for each input power, we will only allow ourselves to apply thresholds such that all $10^7$ measured samples pass the test.

In Fig.~\ref{fig:QRNG_dimensionality}, the certified minimum photon number $n_{R}^{-}$ in mode R is plotted against the input optical power for various security parameters $\epsilon_{\mathrm{fail}}$. The input power was scanned across the linear range of the balanced detector, with the voltage thresholds ($v_{i_{\pm}}^{\pm}$) at each power setting constrained such that all samples passed the test $\mathcal{P}$. Under these constraints, we chose a voltage threshold within the range $0\,\si{\milli\volt}$ to $39.2\,\si{\milli\volt}$. As can be seen, the certified photon number scales linearly with the input power and vanishes for sufficiently small or large photonic inputs. For small powers, $n_{R}^{-}$ goes to zero as no positive solution for Eq.~(\ref{esecrealmain}) with the required $\epsilon_-$ can be found. This is as expected given that, when a low photon number impinges onto detector C, one cannot discern the produced voltage from the detector's inherent electronic noise. Alternatively, for large powers, one can easily achieve a small value for $\epsilon_-$ but it now is not possible to obtain a value of $\epsilon_+$ such that the total certification is valid for $\epsilon_{\mathrm{fail}}$. This is also to be expected as one approaches the balanced detector's saturating power. Finally, for increasing security (i.e. smaller $\epsilon_{\mathrm{fail}}$), $n_R^-$ decreases for a given input power and remains positive over a smaller range of inputs. Indeed, the penultimate data point is non-zero only for $\epsilon_{\mathrm{fail}}\geq 10^{-20}$ and no photon number can be certified with any security for the final point.

\begin{figure}[h]
\includegraphics[width=\linewidth]{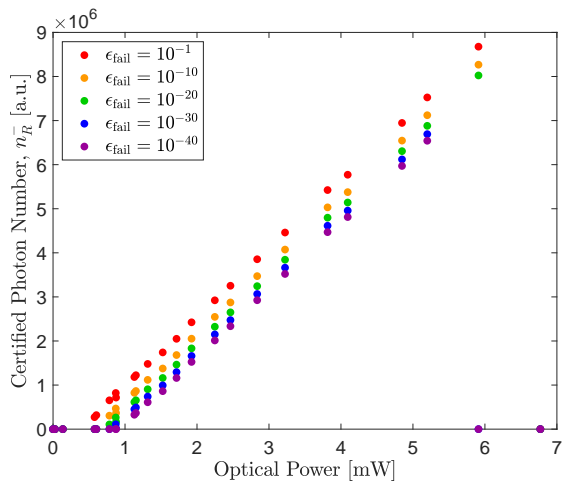}
\caption{Certified minimum photon number $n_{R}^{-}$ in mode R plotted against input optical power for various security parameters $\epsilon_{\mathrm{fail}}$. Voltage thresholds used in the test $\mathcal{P}$ are constrained such that all samples pass.}
\label{fig:QRNG_dimensionality}
\end{figure}

\nw{The main result of this new SDI framework} is shown in Fig.~\ref{fig:min_entropy}, for which a comparison is made between the experimentally estimated min-entropy, various device-dependent (DD) min-entropy models and our SDI approach. The red data points are experimental estimates of the unconditional min-entropy for different average input powers of the laser. These have been calculated from histograms of the difference measurement (shown as inset to Fig.~\ref{fig:min_entropy}) output by the balanced detector. Given these histograms, a Gaussian fit was performed and the retrieved maximum probability $p_{\textnormal{max}}$ was used to estimate the unconditional min-entropy via $H_{\textnormal{min}}=-\log_{2}(p_{\textnormal{max}})$. This corresponds to a naive analysis where all observed fluctuations are assumed to be truly random. The red line is a device-dependent prediction for $H^{\textnormal{DD}}_{\textnormal{min}}(X)$, calculated using our detector model and assuming that the laser is well modelled by a coherent state $\ket{\alpha}$. The resulting curve fits the data well with a coefficient of determination $R^2=98.96\%$, thereby confirming the validity of our modelling. In pink, $H^{\textnormal{DD}}_{\textnormal{min}}(X|E)$ corresponds to the usual device-dependent conditional min-entropy, assuming a known source but accounting for Eve's knowledge of the electronic noise present in our measurement apparatus. As such, it is equal to $H^{\textnormal{DD}}_{\textnormal{min}}(X)$ but shifted down by the min-entropy associated with the electronic noise of the balanced detector. Finally, in green, orange and blue points, we show our SDI model for the certified conditional min-entropy $H^{\textnormal{SDI}}_{\textnormal{min}}(X|E)$ for different values of the security parameter $\epsilon_{\mathrm{fail}}$. These were calculated via Eq.~(\ref{hminthmmain}) using the minimum certified photon numbers $n_{R}^{-}$ displayed in Fig.~\ref{fig:QRNG_dimensionality} for each $\epsilon_{\mathrm{fail}}$. 

\begin{figure}[h]
\includegraphics[width=\linewidth]{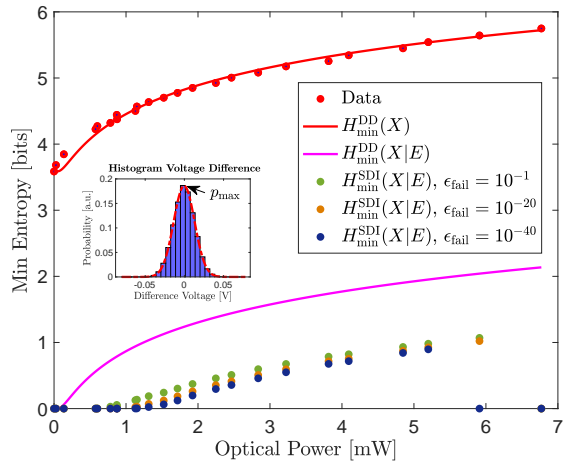}
\caption{Comparison between different min-entropy models. The red data points are the experimentally estimated min-entropies for different optical powers. These are obtained from the difference measurement's voltage histograms shown in the inset (the voltage bins have been artificially thickened by a factor of 10 to make the figure comprehensible). Error bars for the data points have been included with the vertical component arising from the precision of the histogram's Gaussian fit and the horizontal error showing the electronic noise's contribution of detector C when measuring the optical power. $H_{\mathrm{min}}^{\mathrm{DD}}(X)$ (red) and $H_{\mathrm{min}}^{\mathrm{DD}}(X|E)$ (pink) are the device-dependent (DD) min-entropy models unconditioned and conditioned on Eve's knowledge of the noise. $H_{\mathrm{min}}^{\mathrm{SDI}}(X|E)$ (green, orange and blue) are our SDI estimations of the conditional min-entropy plotted against the input optical powers and for various security parameters $\epsilon_{\mathrm{fail}}$.}
\label{fig:min_entropy}
\end{figure}

When comparing the different min-entropies in Fig.~\ref{fig:min_entropy}, it is clear that the claimed level of randomness critically depends on what assumptions are made about the QRG. Indeed, if one were to naively take $H^{\textnormal{DD}}_{\textnormal{min}}(X)$ as a consistent min-entropy model, the QRG's output would consequently be predictable since the electronic noise can be accessible to Eve. On the other hand, whilst $H^{\textnormal{DD}}_{\textnormal{min}}(X|E)$ correctly removes such classical side information, it nevertheless is a device-dependent model for which the experimentalist must trust the proper working of the entire setup, having carefully modelled it and its possible deviations. This means that such scheme must be secure against all sorts of complicated attacks from Eve. In the canonical setup of Fig.~\ref{fig:optical_setup}, a key origin of experimental complexity arises from the input light source. Our approach provides total independence from such complexity whilst still certifying a substantial amount of min-entropy per measurement as well as an explicit quantification of its confidence given by $\epsilon_{\mathrm{fail}}$. As can be seen in Fig.~\ref{fig:min_entropy}, we certify up to $\approx 1.1$ bit of min-entropy with $\epsilon_{\mathrm{fail}}= 10^{-20}$ for the penultimate data point. While this value is about half of what $H^{\textnormal{DD}}_{\textnormal{min}}(X|E)$ predicts, we argue that such compromise is reasonable given that we can still achieve large randomness bit rates for the added SDI security. Indeed, the importance of our SDI protocol's security is starkly illustrated by the final and initial input powers for which no min-entropy is assigned as opposed to the device-dependent model $H^{\textnormal{DD}}_{\textnormal{min}}(X|E)$. 


\subsection{\nw{Real-Time Random Number Extraction}}
\nw{
The real-time extraction of random numbers is performed with a dedicated post-processing Printed Circuit Board (PCB) whose content and functioning are both thoroughly detailed in Appendix \ref{sec:FPGA_details}. Here, instead of using an oscilloscope to read the measurements output by the various detectors in the setup, voltage signals are directly fed to a $b=12\,$bits bit-depth ADC (Analog Devices AD9625) capable of measuring analog inputs up to $3.2\,\si{\giga\hertz}$ with a sampling rate of $F_{S}=2.5\,$GS/s as well as a large ENOB of $\Delta_{\mathrm{ADC}}=9.2\,$bits. This represents a substantial improvement with respect to the ADC found in the oscilloscope used in the characterisation measurements in the previous section. 

As a general principle, to maximise a QRNG's final bit rate, it is important to use an ADC whose ENOB over bit-depth ratio $\Delta_{\mathrm{ADC}}/b$ is as large as possible for a given bit-depth $b$. Indeed, for a fixed number of photons input to the randomness generation measurement, a large ENOB $\Delta_{\mathrm{ADC}}$ allows one to maximise the extractable certified min-entropy per sample $\kappa/m$ since the noise contribution intrinsic to the ADC would be minimised. As explained in Section \ref{sec:hash}, the min-entropy in turn sets the upper limit to the compression ratio, $r \leq \kappa/mb$. Although the ENOB is often not taken into account, this argument makes it clear why one should maximise $\Delta_{\mathrm{ADC}}/b$ rather than solely $b$. Finally, the output of the ADC is sent directly to the FPGA (Zynq Ultrascale$+$ ZU9EG) in order to carry out hashing.}

\nw{The real-time hashing of raw data was implemented using the concurrent pipeline algorithm based on Toeplitz matrix hashing \cite{zhang2016fpga}. The idea of the algorithm is to improve the speed of post-processing by decomposing the large Toeplitz matrix of size $h\times l$ into several submatrices of dimension $k\times l$ and then simultaneously performing matrix multiplication with the raw data. The crucial task of determining $k$, the number of rows for the submatrices, is explained in Appendix \ref{sec:FPGA_details}.  

To demonstrate our protocol, we ran a real-time random number extraction experiment in two distinct configurations producing either long or short strings. These address different real-world applications such as large scale simulations (e.g. Monte Carlo) for which Gb of random numbers are required and standard cryptographic protocols (e.g. Advanced Encryption Standard) typically employing random seeds of kb lengths. The parameters of both configurations are summarised in Table \ref{tab:experiment_FPGA_parameters}.

\begin{table}[h!]
\nw{
\begin{tabular}{|c|c|c|c|}
\hline\xrowht[()]{7pt}
 & Parameters & \multicolumn{2}{c|}{Value} \\
  \hline
  \xrowht[()]{7pt}
  $N_{S}$ & Number of output strings & $1$ & $1.9375\times10^{6}$ \\ 
  \xrowht[()]{7pt}
  $h$ & Hashing block size & $9600\,$bits & $9600\,$bits\\
  \xrowht[()]{7pt}
  $t$ & Hashes per string & $1.9375\times10^{6}$ & $1$ \\
  \xrowht[()]{7pt}
  $m$ & Samples per hash & $800$ & $800$ \\
  \xrowht[()]{7pt}
  $\kappa/m$ & Min-entropy per sample & $5.32\,$bits & $5.34\,$bits \\
  \xrowht[()]{7pt}
  $l$ & Hashing output length & $4155\,$bits & $4210\,$bits\\
  \xrowht[()]{7pt}
  $\epsilon_{\mathrm{fail}}$ & Sample failure $p$ & $1.6\times 10^{-19}$ & $1.1\times 10^{-10}$\\
\xrowht[()]{7pt}
$\epsilon_{\mathrm{hash}}$ & Hashing failure $p$ & $9.0 \times 10^{-17}$ & $3.8\times 10^{-10}$ \\
\xrowht[()]{7pt}
$\epsilon_l$ & Single hashing failure $p$  & $2.2\times 10^{-16}$ & $4.8\times 10^{-10}$  \\
\xrowht[()]{7pt}
$\epsilon$ & Total failure $p$ & $4.3\times 10^{-10}$ & $4.8\times 10^{-10}$\\
\xrowht[()]{7pt}
$R_d$ & Data limited bit rate& $8.05\,$Gb/s & $8.16\,$Gb/s \\
\xrowht[()]{7pt}
$\EV{R}$ & Average bit rate & $8.01\,$Gb/s & $8.16\,$Gb/s \\
\xrowht[()]{7pt}
$L$ & $\epsilon$-random bits per string & $8.05\,$Gb & $4.21\,$kb \\
\hline
\end{tabular}
}
\caption{\nw{Parameters and associated values for the two real-time random number extraction scenarii implemented here.}}
\label{tab:experiment_FPGA_parameters}
\end{table}

For the first configuration, we inserted an optimal input optical power of $5.8\,\si{\milli\watt}$ prior to the randomness generation measurement. The optimisation was performed such that the entire data would pass the certification test $\mathcal{P}$ with a probability $1-\epsilon_{c}=99.5\%$. This yields a certified min-entropy of $H^{\textnormal{SDI}}_{\textnormal{min}}(X|E)=5.32\,$bits per sample acquired by the ADC with a security parameter $\epsilon_{\mathrm{fail}} = 1.6 \times 10^{-19}$. Next, we downsampled the digitised output of the ADC to $1.55\,$GS/s in order to remove any time correlation. This stream of bits was then fed to the FPGA for which the hashing algorithm described above was performed at a speed of $R_{\mathrm{hash}}=193.75\,\si{\mega\hertz}$ and with a Toeplitz matrix of size $h=9600\,$bits and $l=4155\,$bits. We thus achieved a total bit rate of $R_{d} = R_{\mathrm{data}} \times r = 12 \times 1.55 \times 10^{9} \times \frac{4155}{9600} = 8.05\,$Gb/s with an overall composable security of $\epsilon = 4.3\times 10^{-10}$, thereby generating in real-time $N_{S}=1$ string of length $L=8.05\times10^{9}\,$ certified and composably secure quantum random numbers made of $t=1.9375\times10
^{6}$ concatenations. Note that given the probability of passing the test, this obtained bit rate corresponds to a bit rate of $\EV{R} = (1-\epsilon_{c})\times R_{d} = 8.01\,$Gb/s averaged over many runs and with the same level of security. In the second configuration, we took the inverse approach and avoided any concatenation (i.e. $t=1$), allowing for a larger hashing output length of $l=4210\,$bits. Every second, this resulted in $N_{S}=1.9375\times10
^{6}$ strings of length $L=4.21\,$kb each with a composable security of $\epsilon=4.8\times10^{-10}$. The obtained bit rate was thus $R_{d}=8.16\,$Gb/s with the same corresponding average bit rate $\EV{R}=8.16\,$Gb/s up to two decimal places. The numbers obtained from both settings were ultimately found to successfully pass the battery of NIST tests \cite{rukhin2001statistical}. 




This achieves an ultrafast and highly composably secure QRNG based on commercially available components and entirely independent of the incoming light source for which the random numbers are both composably certified and extracted in real-time. To our knowledge, this is the fastest composably secure QRNG (including device-dependent implementations) ever reported.}

\section{Discussion \label{discussion}}
We now return to the desiderata previously outlined for evaluating the usefulness of a QRNG device, namely, level of security, performance (achievable bit rate) and practicality (ease of implementation, durability, and cost). Our protocol used cheap and robust off-the-shelf components that lend themselves to prolonged, high-speed usage and would be amenable to miniaturisation in an integrated photonic architecture. \nw{Utilising an FPGA, we were able to implement the necessary hashing operations in real-time by using the pipeline algorithm of \cite{zhang2016fpga} detailed in Appendix \ref{sec:FPGA_details}. Moreoever, we hashed relatively large blocks which allowed us to extract random numbers at close to the optimal possible rate given the randomness source.}   



Another consideration when developing a protocol for certified randomness is whether such a protocol is composably secure \cite{renner2008security,portmann2014cryptographic}. That is, whether the output of the protocol can then be used as an input to other cryptographic protocols without compromising the security. For example, it can be input to a randomness extractor along with a seed to achieve certified randomness expansion using well known techniques \cite{frauchiger2013true,tomamichel2011leftover}. \nw{Very few implementations enjoy such composable security proofs in either the device-dependent \cite{mitchell2015strong,Haw:2015kx,Gehring:2018wc} or partially device-independent case \cite{cao2016source}. Whilst there is a device-independent result that produces random strings that may be composed \cite{liu2018device},} it is still unknown whether fully device-independent protocols are secure under composition of devices without extra assumptions, e.g. devices are memoryless \cite{barrett2013memory}. It is thus necessary for the moment to move beyond device independence if one desires a fully composably secure protocol. 

In terms of security and performance, our work considers completely general quantum attacks and achieves significantly higher bit rates for a given security parameter than the fastest known source- ($5\,$kb/s in \cite{cao2016source}), measurement- ($5.7\,$kb/s in \cite{nie2016experimental}), semi- ($16.5\,$Mb/s in \cite{brask2017megahertz}) or fully device-independent protocols ($180\,$b/s in \cite{liu2018device}). The only directly comparable work which offers a source-independent composable security proof is \cite{cao2016source}, whose randomness generation rate we improve upon by more than 6 orders of magnitude. \nw{In fact, our work achieves the highest composably secure bit rate for any level of device assumptions, including the fastest device-dependent implementations \cite{Gehring:2018wc}.} 

The experimental architectures most similar to ours are a recent series of papers that involve homodyning the vacuum \cite{marangon2017source}, or squeezed state \cite{michel2019real}, or dual-homodyning the vacuum \cite{avesani2018source} and were claimed to be SDI. Indeed, these works also achieve impressive rates as high as $17\,$Gb/s. To derive a SDI proof, these works apply entropic uncertainty relations \cite{furrer2014position,furrer2014reverse} that can, in principle, lead to devices for which randomness can be certified even if the source of quantum states is completely unknown, provided the measurements acting on these states are well-characterised. However, for realistic homodyne detectors with finite range, the corresponding uncertainty relation becomes trivial and no randomness can be certified \cite{furrer2014position}. Even in the case of infinite range detectors, the modelling of a photon difference as a quadrature measurement is only valid in the case where the input photon is small with respect to the local oscillator. This problem can be ameliorated but only at the price of introducing an energy assumption (similar to the semi-device-independent approach) upon the source, thus jeopardising the claimed source independence. 

\nw{A final technical point is that, although the importance of considering digitisation noise via the ENOB of the ADC has been pointed out previously \cite{zhang2016fpga, marangon2017source}, many experiments fail to take this into account. This key consideration has the effect of reducing the retrievable min-entropy per sample, thereby considerably lowering the bit rates reported in the vast majority of the corresponding literature. A comparison of the security, assumptions and performance of a selection of other works compared to ours can be found in Table \ref{tab:comparison_works}.}

\begin{table}[h!]

\nw{\begin{tabular}{|c|c|c|c|c|c|}
	\hline\xrowht[()]{10pt}
Work & $\begin{array}{c} \textrm{Trust} \\ \textrm{level} \end{array}$ & $\begin{array}{c} \textrm{Use of} \\ \textrm{ENOB} \end{array}$ & $\epsilon$  & $\begin{array}{c} \textrm{QRG} \\ \textrm{bit rate }\\
\textrm{[Mb/s]}\end{array}$ & $\begin{array}{c} \textrm{QRNG} \\ \textrm{bit rate}\\
\textrm{[Mb/s]}\end{array}$ \\
  \hline
  \hline\xrowht[()]{10pt}
\cite{Gehring:2018wc} & DD & No &$10^{-10}$& 10740 & 8000 \\
\hline\xrowht[()]{10pt}
\cite{marangon2017source} & sSDI & Yes & -- & 1700 & -- \\
\hline\xrowht[()]{10pt}
\cite{michel2019real} & sSDI & No & -- & 0.0082 & -- \\
\hline\xrowht[()]{10pt}
\cite{brask2017megahertz} & sDI & N/A & -- & 16.5 & --\\
\hline\xrowht[()]{10pt}
\cite{cao2016source} &  SI & N/A & $ 10^{-15}$& 0.005 & -- \\
  \hline\xrowht[()]{10pt}
  \cite{liu2018device} &  DI & N/A & $ {}^{*}10^{-5}$& 0.000181 &--\\
  \hline\xrowht[()]{10pt}
$\begin{array}{c} \textrm{This} \\ \textrm{work} \end{array}$ & SDI & Yes & $10^{-10}$ & 8211 & 8050 \\
\hline
\end{tabular}
\caption{Comparison of randomness generation protocols. DD: device-dependent; sSDI: semi-source device independent; DI: device independent; SDI: source-device independent. ${}^{*}$Not proven secure under composition of devices. \label{tab:comparison_works}}}
\end{table}

Finally, we turn to a quantitative comparison between this work and earlier protocols based on homodyne detection in the device-dependent \cite{Haw:2015kx,Gehring:2018wc} and semi-SDI contexts \cite{marangon2017source,avesani2018source,michel2019real}. Strictly speaking, direct comparison with the semi-SDI protocols is impossible since these fail to give a composable security parameter. \nw{Also, in practice the achievable rates depend heavily on many technical constraints such as the detector noise and especially the effective number of ADC bits.} In Fig.~\ref{fig:comparison_our_work_vs_homodyne}, we consider a simpler calculation of the min-entropy generated in a single round using ideal equipment to compare the ultimate rates of these different protocols. \nw{The security parameter for the displayed SDI curves is chosen to be $\epsilon_{\mathrm{fail}} = 10^{-10}$ with the honest passing probability chosen as $1-\epsilon_C = 0.995$. For the EUR protocol, the probability of making a randomness generating measurement was set to be $p_X = 0.9$ and the photon number of the local oscillator used in the homodyne detection was $n_{\mathrm{LO}} = 10^7$. Details of the calculations are give in Appendix \ref{comp}.}

\begin{figure}[h!]
\includegraphics[width=\linewidth]{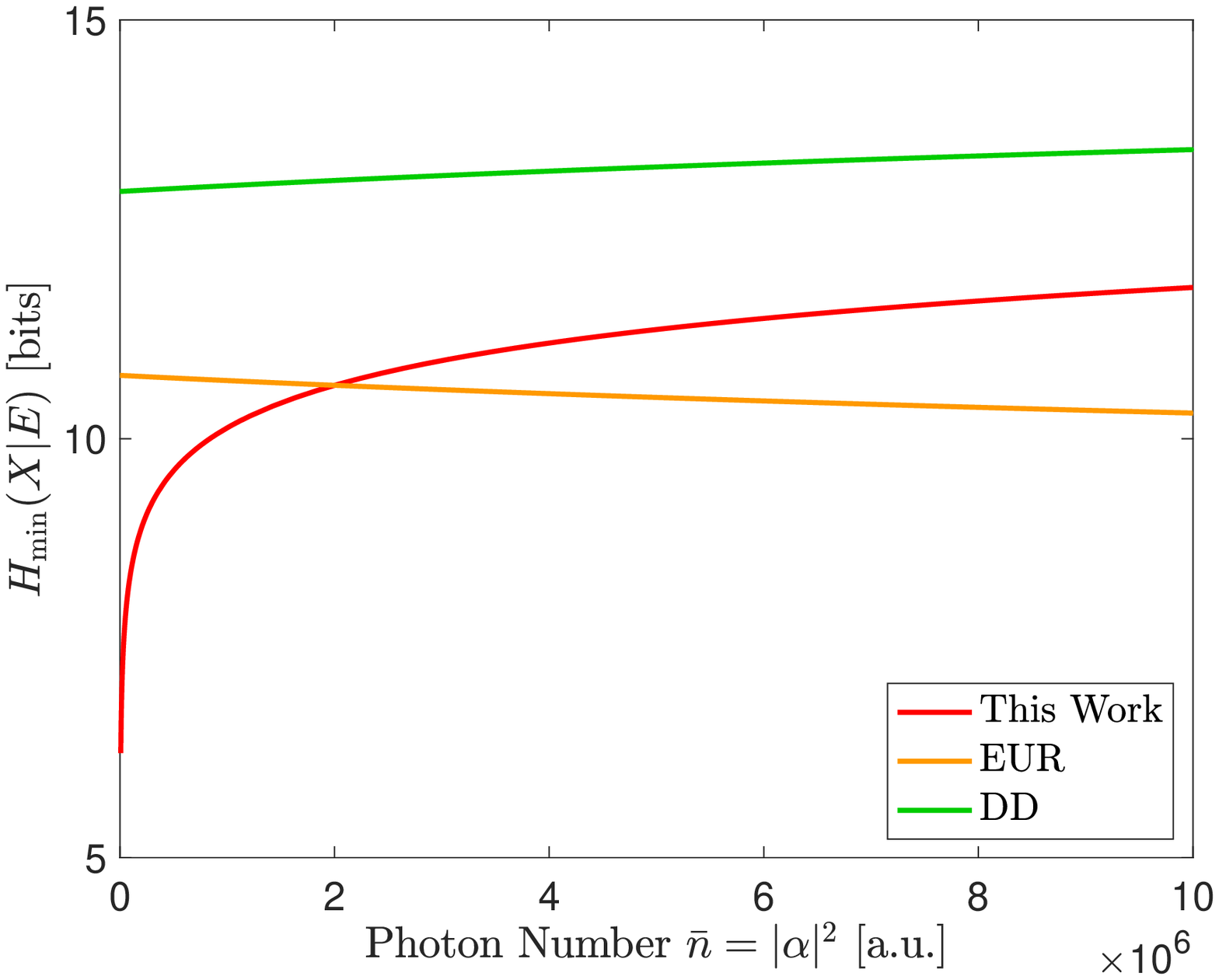}
\includegraphics[width=\linewidth]{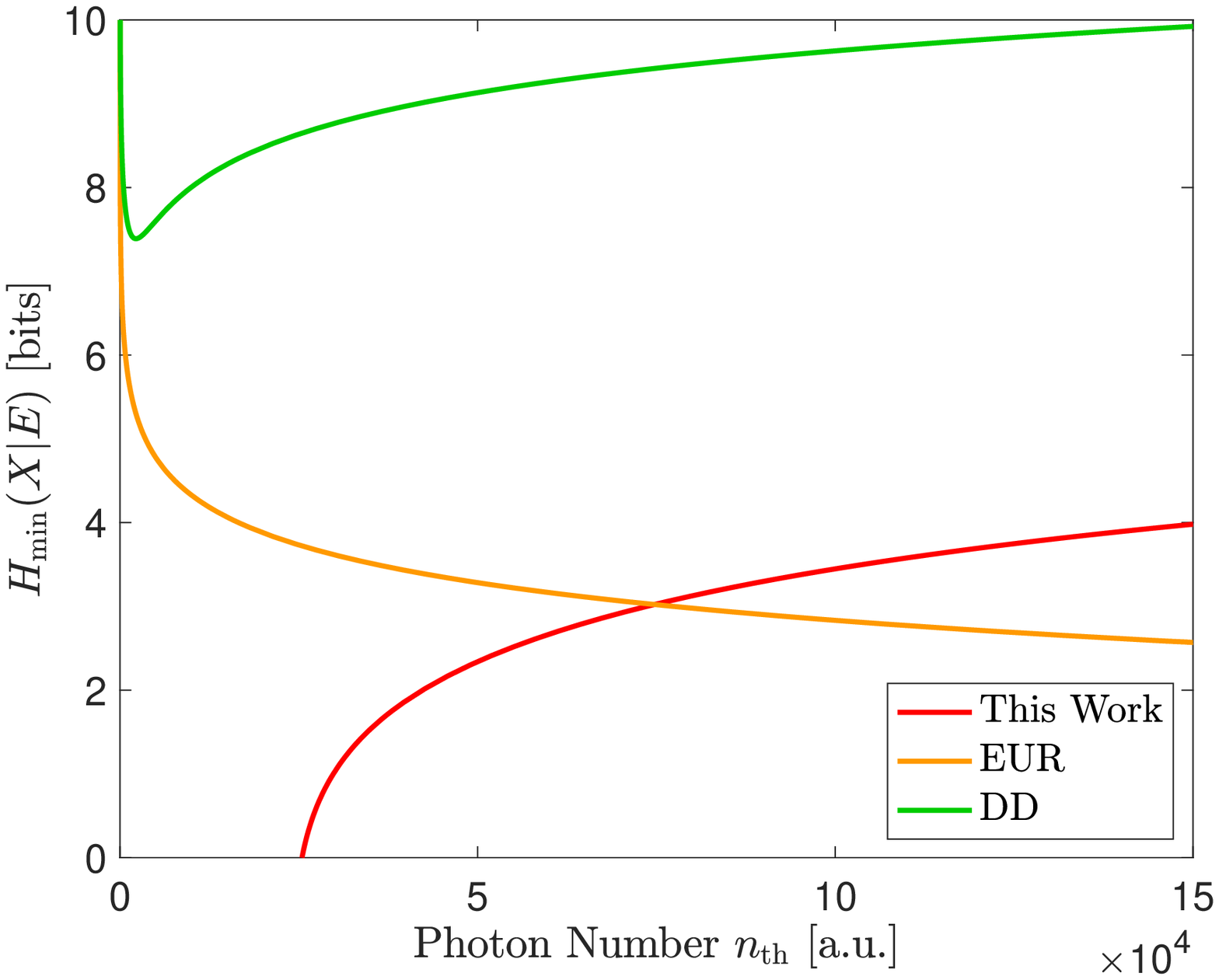}
\caption{Comparison between the min-entropy per round for the present SDI protocol, those based on homodyning the vacuum state using an entropic uncertainty relation (EUR) and device-dependent (DD) homodyning for different inputs states as a function of mean photon number. Top: coherent state; bottom: thermal state.}
\label{fig:comparison_our_work_vs_homodyne}
\end{figure}

For certain input states we identify fundamentally different scalings in some instances. Although we actually consider upper bounds on the rates for the device-dependent and semi-SDI schemes, thereby penalising this work by comparison, we see dramatically different scalings between this work and the semi-SDI homodyne scheme. As can be observed in Fig.~\ref{fig:comparison_our_work_vs_homodyne}, if the input state is one half an entangled two-mode squeezed vacuum state (i.e. a thermal state) or a coherent state, then the randomness of homodyne protocols decreases as function of the photon number of the input state whereas the randomness of the present protocol monotonically increases. For sufficiently large photon numbers, this work scales identically to the device-dependent case, thereby achieving significantly improved security with only a constant factor reduction in performance. \nw{Moreover, it should be noticed that for an input coherent state, the photon number from which this work's generated min-entropy surpasses that obtained from the EUR protocol is relatively small (i.e. $\bar{n}=|\alpha|^{2}\approx 2\times10^{6}$). This crossing point and the ensuing advantageous scaling make this work even more desirable from a realisation point of view since it occurs for a coherent state, the most practical and hence widely utilised state in experimental quantum optics.} Overall, these key considerations highlight the fundamental quantitative differences in between this work and traditional homodyne based protocols.

\section{Conclusion}
In summary, we presented and experimentally implemented a SDI protocol based on the quantum nature of untrusted light. \nw{Our QRNG achieves both state-of-the-art ultrafast randomness generation and real-time random number extraction with a bit rate of $R_{d}=8.05\,$Gb/s whilst providing a rigorous and specific security parameter of $\epsilon = 10^{-10}$ for the generated random numbers with no assumptions on the light source. There are several avenues for improvement.} A higher bandwidth balanced detector for the randomness generation speed as well as a \nw{larger effective bit-resolution} of the ADC for the retrievable min-entropy per sample are primary examples among them. \nw{Lastly, the present configuration could be upgraded by connecting more randomness sources (say $\gamma > 1$ of such sources) to the same FPGA and carrying out parallel real-time post-processing. This would achieve an unparalleled average QRNG bit rate of $\gamma \times \EV{R}$ for the same level of composable security.}

\section{Acknowledgements}
This work was supported by funding from the UK Engineering and Physical Sciences Research Council (EPSRC) National Quantum Technology Hub in Networked Quantum Information Technologies (NQIT). NW acknowledges funding from the European Union's Horizon 2020 research and innovation programme under the Marie Sklodowska-Curie grant agreement No. 750905. A.K.F acknowledges RBFR grant No. 18-37-20033.

\newpage
\clearpage

\appendix
\onecolumngrid

\section{Certifiable randomness of ideal difference measurement} \label{ideal}
To begin with, consider the randomness generation measurement of Fig.~\ref{fig:idea}. It consists of a beam splitter BS$_{0}$ with reflectivity $r_{0}=\frac{1}{2}$, an input mode R, a trusted vacuum fed into the other input mode and two output photodetectors A and B performing a difference measurement. It simplifies matters greatly if we can prove that the potential eavesdropper in charge of our photonic source is making definite photon number states (i.e. Fock states) for each round of the protocol. In particular, we would like to rule out any sophisticated, collective strategy where Eve sends a complicated state that is entangled across all rounds of the protocol. 

Intuitively, this should be the case because the randomness generation measurement for each round is a photon number difference and can be thought of as a coarse graining over an initial measurement that is diagonal in the Fock basis. Here, this is shown by writing out the POVM directly and the optimality of unentangled Fock state inputs from Eve's perspective becomes explicit.  

For a single round, the entire process of mixing $\hat{\rho}_R$ with a vacuum ancilla $\ket{0} \in \mathcal{H}_V$ and then making Fock state projections upon both output ports can be seen as a POVM on $\mathcal{H}_R$, the Hilbert space of $\hat{\rho}_R$. Consider the probability for detecting $n_A$ and $n_B$ photons at detectors A and B. This is given by
\eqn{p(n_A,n_B) \ee \mathrm{tr} \left \{ \hat{U}_{BS_{0}} (\hat{\rho}_R \otimes \ket{0}\bra{0}) \hat{U}^{\dag}_{BS_{0}}(\ket{n_A}\ket{n_B} \bra{n_A}\bra{n_B})  \right \} \nn\\
\ee \mathrm{tr}_R \left \{ \mathrm{tr}_V \left \{  (\hat{\rho}_R \otimes \ket{0}\bra{0}) \hat{U}^{\dag}_{BS_{0}}(\ket{n_A}\ket{n_B} \bra{n_A}\bra{n_B}) \hat{U}_{BS_{0}} \right \} \right \} \nn\\
\ee \mathrm{tr}_R \left \{ \hat{\rho}_R \hat{M}(n_A,n_B) \right \} \,,}
where \eqn{\hat{M}(n_A,n_B) =  \bra{0}\hat{U}^{\dag}_{BS_{0}}\ket{n_A}\ket{n_B} \bra{n_A}\bra{n_B} \hat{U}_{BS_{0}}\ket{0}\label{M2}\,,}
is the corresponding POVM element in the input state Hilbert space (with the subscript R suppressed for brevity). This expression is just the evolution of the Fock state projections back through the beam splitter BS$_{0}$ and projected onto the vacuum ancilla. To get an explicit expression, it is simpler to switch to the Heisenberg picture for the reverse beam splitter transformation
\eqn{ \ket{n_A}\ket{n_B} \ee\frac{(\adag_A)^{n_A}}{\sqrt{n_A}!}\frac{(\adag_B)^{n_B}}{\sqrt{n_B}!}\ket{0} \nn \\
&\stackrel{U_{BS_{0}}^\dag}{\mapsto}& \frac{\bk{\frac{\adag_E + \adag_V}{\sqrt{2}}}^{n_A}}{\sqrt{n_A}!}\frac{\bk{\frac{\adag_E - \adag_V}{\sqrt{2}}}^{n_B}}{\sqrt{n_B}!}\ket{0} \nn\\
\ee \frac{\sum_{k=0}^{n_A} \sum_{j=0}^{n_B}(\adag_E)^{n_A-k}(\adag_V)^k \binom{n_A}{k}(-1)^j(\adag_E)^{n_B-j}(\adag_V)^j \binom{n_B}{j}}{2^{(n_A+n_B)/2}\sqrt{n_A!n_B!}}  \ket{0} \nn \\
\ee \frac{\sum_{k=0}^{n_A} \sum_{j=0}^{n_B}\sqrt{(n_A+n_B-j-k)!(j+k)!} \binom{n_A}{k}(-1)^j\binom{n_B}{j}}{2^{(n_A+n_B)/2}\sqrt{n_A!n_B!}}  \ket{n_A+n_B-j-k}_{R}\ket{j+k}_{V} \,.}

Acting on the left with $\bra{0}$ on the ancilla mode implies that we must have $j+k = j = k =0$, thus
\eqn{\bra{0}\hat{U}_{BS_{0}}^{\dag} \ket{n_A}\ket{n_B} \ee \frac{\sqrt{(n_A+n_B)!}}{2^{(n_A+n_B)/2}\sqrt{n_A!n_B!}} \ket{n_A+n_B}_{R} \,,}
and hence
\eqn{\hat{M}(n_A,n_B) \ee \frac{(n_A+n_B)!}{2^{(n_A+n_B)}n_A!n_B!} \ket{n_A+n_B}\bra{n_A+n_B}_{R} \nn\\
\ee 2^{-N}\frac{N!}{n_A!(N-n_A)!} \ket{N}\bra{N}_{R} \label{n1n2} \,,} 
where we have substituted in the total photon number $N := n_A + n_B$. As expected, each POVM element is proportional to a single Fock state of fixed photon number $N$ and the coefficient can be understood intuitively. Indeed, each of the $N$ photons can be thought of as individually randomising at the beam splitter. The probability for a specific sequence of paths taken by each photon is $2^{-N}$ and thus the probability of observing the POVM element $\hat{M}(n_A,n_B)$ is the number of paths such that $n_A$ out of $N$ photons could have been recorded at detector A, which is $\binom{N}{n_A}$ as above. 

If we consider the sum measurement, it is just a coarse graining over the two outcome POVM, summing together all the elements such that $n_A+n_B = N$. The POVM elements of the sum measurement $\mathbb{Z} = \{ \hat{Z}(N)\}$ are
\eqn{\hat{Z}(N) = \sum_{n_A=0}^N \hat{M}(n_A,N-n_A)\,.}

Using the fact that $\sum_{k=0}^n \binom{n}{k} = 2^n$, we can see that $\hat{Z}(N) = \ket{N}\bra{N}_{R}$ and it is thus just a photon number projector as expected.

The randomness generation measurement is another coarse graining. However, it will turn out to have larger rank and consequently some randomness for all possible input states other than the vacuum. Define $\mathbb{X} = \{\hat{X}(x)\}$ as the POVM elements of the randomness generation measurement corresponding to the cases where $n_A-n_B := x$. These are given by
\eqn{\hat{X}(x) \ee \sum_{n_A = x}^\infty \hat{M}(n_A,n_A-x) \nn \\
\ee \sum_{n_A = x}^\infty 2^{-(2n_A-x)}\binom{2n_A-x}{n_A} \ket{2n_A-x}\bra{2n_A-x}_{R}\,,}
if $x$ is positive and
\eqn{\hat{X}(x) \ee \sum_{n_A = |x|}^\infty \hat{M}(n_A-|x|,n_A) \nn \\
\ee \sum_{n_A = |x|}^\infty 2^{-(2n_A-|x|)}\binom{2n_A-|x|}{n_A} \ket{2n_A-|x|}\bra{2n_A-|x|}_{R}\,,}
if $x$ is negative or 
\eqn{\hat{X}(x) \ee \sum_{n_A = |x|}^\infty 2^{-(2n_A-|x|)}\binom{2n_A-|x|}{n_A} \ket{2n_A-|x|}\bra{2n_A-|x|}_{R}\label{x}\,,}
for all $x$. 

Note that for $x$ even (odd), then $\hat{X}(x)$ only has support over even (odd) number states. Clearly, if Eve inputs a vacuum state, then the difference outcome can be predicted with certainty as $x=0$. However, as pointed out in the main text, if Alice observes a value $N$ for her sum measurement, then regardless of the original input, she performs a projection onto the state $\ket{N}$ and can immediately calculate the guessing probability of the $\mathbb{X}$ measurement $p_{\mathrm{guess}} =\max_x \bra{N}\hat{X}(x)\ket{N}$ from Eq.~(\ref{x}) and hence the associated min-entropy. For perfect measurements, this would guarantee the min-entropy with certainty and in a SDI manner.

Now, consider the full setup shown in Fig.~\ref{fig:idea}. We introduce the certification measurement in mode C which is done by tapping off a fraction of the completely unknown incoming light in mode E with a beam splitter BS$_{1}$ of reflectivity $r_{1}$. The input state $\hat{\rho}_{E}$ is mixed with vacuum on BS$_{1}$ and the reflected beam in mode C is measured at detector C while the transmitted beam in mode R is input to the randomness generation measurement. For simplicity, we will imagine that the outcome at detector $C$ is also always given to Eve. Writing the photon number projections as operators on the input Hilbert space $\mathcal{H}_{E}$ is the same calculation as Eq.~(\ref{n1n2}), except now with a beam splitter of reflectivity $r_{1}$ instead of $\frac{1}{2}$. This gives 
\eqn{\hat{M}(n_C,n_R) \ee \frac{r_{1}^{n_C}(1-r_{1})^{n_R}(n_C+n_R)!}{n_C!n_R!}\ket{n_C+n_R}\bra{n_C+n_R}_{E} \label{cert2} \,,}
and hence the certification measurement has elements
\eqn{\hat{N}_C(n_C) \ee \sum_{n_R=0}^\infty \frac{r_{1}^{n_C}(1-r_{1})^{n_R}(n_C+n_R)!}{n_C!n_R!}\ket{n_C+n_R}\bra{n_C+n_R}_{E} \label{Ncert}\,.}

Given this measurement, one cannot exactly determine the number of photons in mode R incident onto the randomising beam splitter BS$_{0}$, but one can obtain a lower bound on the min-entropy of $m$ such measurements except with some failure probability $\epsilon_{\mathrm{fail,m}}$. Specifically, we impose a test $\mathcal{P}$ at detector C which is passed if the measured photon number is greater than a lower threshold $n_{C}^{-}$. 

Upon passing the test $\mathcal{P}$, we certify a lower bound $n_{R}^{-}$ on the photon number in mode R impinging onto the randomness generation measurement. We formally state and prove this result below.

\begin{Theorem} \label{rgendideal} 
An optical setup consisting of 
\begin{itemize}
\item Two trusted vacuum modes
\item Two beam splitters of reflectivity $r_{0}=\frac{1}{2}$ and $r_{1}$
\item Three ideal photon counting detectors A, B and C
\end{itemize}

utilised to perform a certification measurement modelled by Eq.~(\ref{Ncert}) with lower threshold $n_{C}^{-}$ and a randomness generation measurement modelled by Eq.~(\ref{x}) can be used as a certified (m,$\kappa$,$\epsilon_{\mathrm{fail,m}}$,$\epsilon_c$)-randomness generation protocol as per Definition \ref{QRGdef} without making any assumptions about the photonic source with
\eqn{\kappa &\geq& - m\log_2 \left ( 2^{-n_{R}^{-}}\binom{n_{R}^{-}}{ \left \lfloor \frac{n_{R}^{-}}{2} \right \rfloor} \right ) \nn \\
    &\geq& m\bk{\frac{1}{2} \log_2\bk{\half \pi n_{R}^{-}} - \mathcal{O}\bk{\frac{1}{n_{R}^{-}}}}\label{hminthmideal} \,,}

\eqn{ \epsilon_{\mathrm{fail,m}} &\leq& m\exp\bk{-\frac{2 (r_{1}  (n_{R}^{-}+n_{C}^{-}-1)-n_{C}^{-})^2}{n_{R}^{-}+n_{C}^{-}-1}} \label{esec} \,,} 
and \eqn{\epsilon_c = 1 - e^{-|\alpha|^2}\sum_{n=0}^{\infty}\sum_{n_C=n_{C}^{-}}^\infty \frac{|\alpha|^{2n}}{n!}\frac{r_{1}^{n_C}(1-r_{1})^{n-n_C}n!}{n_C!(n-n_C)!}   \label{completeA} \,,} 
using a coherent state $\ket{\alpha}$ as an input.
\end{Theorem}

\begin{proof}

{\bf Security:}
The key feature here is the diagonal nature in the photon number basis of all measurements performed in the protocol. We first prove a Lemma regarding such measurements.

\begin{Lemma} \label{fockopt}
For a $m$-round, SDI protocol involving a measurement $\mathbb{Q}$ in each round that is diagonal in the number basis with elements
\eqn{\hat{Q}(q) = \sum_{n} c_n(q) \ket{n}\bra{n}, \hspace{2mm}\sum_q \hat{Q}(q) = \mathbb{I}\label{diagmeas}\,,} 

Eve's optimal strategy to maximise the probability of a desired outcome $q^*$ is to input a pure Fock state $\ket{n^*}$ for each round. Moreover, this remains true for inputs with restricted support in the Fock basis. 
\end{Lemma}

\begin{proof}

One way to see this is to consider a diagonalising map in the Fock basis applied to the input of the $i^{th}$ round
\eqn{\hat{\mathcal{D}}_i(\hat{\rho}) = \sum_n \bra{n}\hat{\rho}\ket{n} \ket{n}\bra{n}\,.}

This operator commutes with the $\mathbb{Q}$ measurement and there is no operational way for Eve (or anyone else) to distinguish between directly measuring $\mathbb{Q}$ or measuring $\mathbb{Q}$ after first applying $\hat{\mathcal{D}}$. As such, we could imagine that we are in fact always applying $\hat{\mathcal{D}}$ to each run of the protocol\footnote{That is, the probabilities for any string of measurement outcomes $X^m = [x_1,x_2,..,x_m]$ satisfy $p(X^m) = \mathrm{tr}\{ \hat{\rho}^m_{AE} \otimes_{\nu=1}^m \hat{X}(x_\nu)\} = \mathrm{tr}\{ \hat{\sigma}_{AE}^m\otimes_{\nu=1}^m \hat{X}(x_\nu) \}$ where $\hat{\sigma}^m_{AE} = \otimes_{\nu=1}^m \hat{\sigma}_{\nu}$ with $\hat{\sigma}_\nu = \hat{\mathcal{D}}(\mathrm{tr}_{\bar{\nu}}\{ \hat{\rho}^m_{AE}\})$. Note that $\mathrm{tr}_{\bar{\nu}}$ denotes the trace over all modes except the $\nu^{\mathrm{th}}$ mode.}. To start with, since $\hat{\mathcal{D}}$ satisfies the definition of an entanglement breaking map \cite{horodecki2003entanglement}, we may safely conclude that Eve's optimal strategy will not include any entanglement as there is no way for such entanglement to be noticeable. Moreover, if we consider any individual round of the protocol, we can write its purification as a mode $E'$ held by Eve (including potentially all the other rounds of the protocol) in the Schmidt form $\ket{\Psi_{E'E}} = \sum_j \lambda_j \ket{j}_{E'}\ket{j}_E$ (with $\ket{j}$ not necessarily the Fock basis) and act $\hat{\mathcal{D}}$ upon it. This yields
\eqn{(\mathbb{I}\otimes\hat{\mathcal{D}}) \ket{\Psi_{E'E}}\bra{\Psi_{E'E}} \ee \sum_{j,k} \lambda_j\lambda_k^* \ket{j}\bra{k} \hat{\mathcal{D}}(\ket{j}\bra{k}) \nn \\
\ee \sum_n \hat{\sigma}_{E'_n} \otimes \ket{n}\bra{n}\,,}
where $\hat{\sigma}_{E'_n} = \sum_{j,l,n} \lambda_l\lambda^*_j \braket{n}{l}\braket{j}{n}\ket{l}\bra{j}$. This means that the most general state Eve can effectively prepare for the input mode E is of the form 
\eqn{\hat{\rho}_E = \sum_n p(n) \ket{n}\bra{n}\,,} 
where $p(n) = \sum_j |\lambda_j\braket{n}{j}|^2$. In other words, the input state for each run of the protocol is effectively just a mixture of Fock states (potentially classically correlated between rounds). Intuitively, one would imagine that the best strategy for Eve would be to choose a state such that $\{\ket{j}\}$ is indeed the Fock basis and, moreover, to make $p(n)$ simply a delta function at some fixed $n$.

We can show this as follows. Let $p^*(n)$ be the distribution of the optimal input state that maximises the probability of $q^*$ and $\{c_n(q^*)\}$ be the Fock state coefficients for that element as given in Eq.~(\ref{diagmeas}). Then, Eve's optimal probability is given by 
\eqn{p_{\mathrm{guess}} \ee \mathrm{tr} \{\hat{\rho}_{E'E} (\mathbb{I}\otimes \hat{Q}(q^*)) \} \nn \\
\ee \sum_n p^*(n) c_n  \leq \max_n c_n \times \sum_n p^*(n) =  c_{n^*}\,,}
where we have defined $n^*$ as the value that achieves the maximum. This optimal guessing probability would be saturated by choosing an input state $\ket{n^*}$, therefore the optimal input state is indeed a pure Fock state. 

Note that the result extends straightforwardly to the case where the input state is restricted to have support only over a finite range of number states $[n_{R}^{-}, n_{R}^{+}]$. Let $p^*(n)$ be a probability distribution over $[n_{R}^{-}, n_{R}^{+}]$, $x^*$ be the value of the most likely POVM element of the difference measurement given that input state and $c_n$ be the Fock state coefficients for that element as given in Eq.~(\ref{x}). Then
\eqn{p_{\mathrm{guess}} \ee \mathrm{tr} \{\hat{\rho}_{E'E} (\mathbb{I}\otimes \hat{X}(x^*)) \} \nn \\
\ee \sum_{n_{R}^{-}}^{n_{R}^{+}} p^*(n) c_n  \leq \max_{n\in [n_{R}^{-}, n_{R}^{+}] } c_n \times \sum_n p^*(n) =  c_{n^*}\,.}

Therefore, the optimal input state is $\ket{n}$ with $n\in[n_{R}^{-}, n_{R}^{+}]$. This result can be independently applied to each run of the protocol (by including the other rounds in the purification, Eve has already been granted the option to utilise a sophisticated collective encoding), hence we can conclude that Eve's optimal probability to obtain a string of outcomes for all $m$ rounds is to choose a single Fock state for each round.

\end{proof}

Given Lemma \ref{fockopt}, we now lower bound the min-entropy under the assumption that Eve's input state only has support over number states in the range $[n_{R}^{-},\infty[$. Eve's guess for the difference measurement outcome will always be just the outcome of the most likely element of the difference element defined in Eq.~(\ref{x}). Thus, if we choose $x^*$ to be the most probable outcome of the difference measurement (whatever that might be), then we can immediately conclude that for input states restricted to have support only over the range $[n_{R}^{-},\infty[$, Eve's optimal strategy to maximise the occurrence of $x^*$ (and hence her guessing probability) will be to input a number state $\ket{n} \in [n_{R}^{-},\infty[$. In fact, it will be optimal to input the smallest number state $\ket{n_{R}^{-}}$. We have
\eqn{p_{\mathrm{guess}} \ee \max_{n} \bra{n}\hat{X}(x^*)\ket{n} \nn\\
&\leq& \max_{n \in [n_{R}^{-},\infty[} 2^{-n} { n \choose \left \lfloor \frac{n+|x^*|}{2} \right \rfloor}\nn\\ 
\ee \max_{n \in [n_{R}^{-},\infty[} 2^{-n} { n \choose \left \lfloor \frac{n}{2} \right \rfloor} \nn\\
\ee 2^{-n_{R}^{-}} { n_{R}^{-} \choose \left \lfloor \frac{n_{R}^{-}}{2} \right \rfloor} \,,}
where in the penultimate line, we used the fact that $n\choose k$ is maximal for $k = \left \lfloor \frac{n}{2} \right \rfloor $ and monotonically decreases for greater and smaller values of $k$, which means that the smallest allowed $x$ will be optimal. In the final line, we used that $2^{-n}\binom{n}{\left \lfloor \frac{n+x}{2} \right \rfloor}$ decreases monotonically in $n$. To see this, first note that for $n$ even $\left \lfloor \frac{n+1}{2} \right \rfloor  = \left \lfloor \frac{n}{2} \right \rfloor$ and for $n$ odd $\left \lfloor \frac{n+1}{2} \right \rfloor  = \left \lfloor \frac{n}{2} \right \rfloor +1$. Thus the ratio of successive terms is
\eqn{ \frac{2^{-(n+1)}{n+1\choose \left  \lfloor \frac{n+1}{2} \right \rfloor}}{2^{-n}{n\choose \left \lfloor \frac{n}{2} \right \rfloor}} \ee \half (n+1)\frac{\left \lfloor \frac{n}{2} \right \rfloor!}{\left \lfloor \frac{n+1}{2} \right \rfloor !} \frac{\bk{n-\left \lfloor \frac{n}{2} \right \rfloor}!}{\bk{n+1-\left \lfloor \frac{n+1}{2} \right \rfloor }!} \nn\\
\ee \begin{cases}
 \frac{ 1}{2}(n+1) \frac{\bk{n- \frac{n}{2} }!}{\bk{n+1- \frac{n}{2}  }!} =\half \frac{ (n+1)}{n+1- \frac{n}{2} } = \frac{n+1}{n+2} < 1 \,, &  n \hspace{2mm} \mathrm{even} \\
 \frac{ 1}{2}(n+1)  \frac{\lfloor\frac{n}{2}\rfloor !}{(\lfloor\frac{n}{2}\rfloor + 1)!} =  \half \frac{n+1}{\lfloor\frac{n}{2}\rfloor + 1 } = 1\,, &  n \hspace{2mm} \mathrm{odd}
 \end{cases} \,.}
 
Substituting this optimal guessing probability into the definition of the conditional min-entropy gives the expression in Eq.~(\ref{hminthmideal}). 
 
Now, we show that provided that in each round the certification measurement outcome is above a certain threshold $n_{C}^{-}$, the input to the randomness generation measurement is $\epsilon_{\mathrm{fail,m}}$-indistinguishable from a state with support only over $[n_{R}^{-},\infty[$. The worst case scenario would be that whenever Eve can distinguish the real state from one with restricted support, she learns the full measurement record. We can thus interpret this distinguishing probability as a lower bound to the failure probability for the whole protocol.

Specifically, we are interested in the probability where the certification measurement takes a value which passes our test $\mathcal{P}$ whilst simultaneously a smaller than desired number of photons goes to the randomness generation measurement, thereby representing a failure of the protocol. As such, we introduce a failure operator corresponding to there being $n_{R}^{-}$ or fewer photons in mode R given $n_{C}$ photons in mode C expressed as
\eqn{\hat{F}(n_C,n_{R}^{-}) \ee \sum_{n_R=0}^{n_{R}^{-}} \frac{r_{1}^{n_C}(1-r_{1})^{n_R}(n_C+n_R)!}{n_C!n_R!}\ket{n_C+n_R}\bra{n_C+n_R}_{E} \label{certmeas}\,.}

The failure probability for Eve successfully cheating the test in a single round is then given by
\eqn{\epsilon_{\mathrm{fail}} = \max_{\hat{\rho}_E} \mathrm{tr} \left \{\hat{\rho}_E \sum_{n_C = n_{C}^{-}}^{\infty} \hat{F}(n_C,n_{R}^{-}) \right \} \label{e1failideal} \,.}

It is straightforward to see (and we show it in Appendix \ref{details}) that this probability is also explicitly the probability of passing the test, multiplied by the distinguishing probability between the real input to the randomness measurement, $\hat{\rho}_R$, and the closest state with support solely in the range $[n_{R}^{-},\infty[$ as one would expect in a composably secure framework. Since $\hat{F}$ is once more diagonal in the photon number basis, we can again apply Lemma \ref{fockopt} to conclude that Eve's optimal strategy is achieved by a single number state $\ket{n_E}$. Substitution via Eq.~(\ref{certmeas}) gives
\eqn{\epsilon_{\mathrm{fail}} &\leq& \max_{n_{E}} \sum_{\substack{n_C = \max \{ n_{C}^{-}, \\ n_{E}-(n_{R}^{-}-1) \}}}^{n_{E}}  \frac{r_{1}^{n_C}(1-r_{1})^{n_{E}-n_C}n_{E}!}{n_C!(n_{E}-n_C)!}  \label{e1ideal}\,.}

The lower limit on $n_C$ in the sum comes from the fact that for $n_{E}>n_{C}^{-} + n_{R}^{-} - 1$, the requirement for at least $n_{C}^{-}$ photons at detector C is superseded by the requirement that there be less than $n_{R}^{-}$ photons in mode R which implies $n_C> n_{E} - n_{R}^{-} $. In fact, we show that Eve's optimal input is to send precisely $n_{E}^{\mathrm{opt}}=n_{C}^{-} + n_{R}^{-} - 1$ photons. The summand is a generic binomial distribution
\eqn{\mathcal{B}(r_{1},n_{E},k) = \frac{r_{1}^{k}(1-r_{1})^{n_{E}-k}n_{E}!}{k!(n_{E}-k)!}\,,}
such that the failure probability in Eq.~(\ref{e1ideal}) can be seen as the complement of the binomial cumulative distribution function (CDF). For a fixed lower limit in the sum, the failure probability increases monotonically with $n_{E}$. However, once $n_{E}> n_{C}^{-} + n_{R}^{-} - 1$, the situation is more complicated because the limits of the sum change as well as the summand. Indeed, instead of running from $n_{C}^{-}$ to $n_{E}$, it will run from $n_{C}^{-}+1$ to $n_{E}+1$ as argued above. We now show that the difference between successive terms of the sum for all values $n_{E}$ larger than this threshold is negative and thus the function is monotonically decreasing in $n_{E}$. Hence, it reaches its maximum for $n_{E}^{\mathrm{opt}}=n_{C}^{-} + n_{R}^{-} - 1$. 

For $n_{E}= n_{C}^{-} + n_{R}^{-} - 1$,  we can write $\epsilon_{\mathrm{fail}}$ for the corresponding successive input Fock states as
\eqn{\epsilon_{\mathrm{fail}}(n_{E}+1) - \epsilon_{\mathrm{fail}}(n_{E})  \ee \sum_{n_C = n_{C}^{-}+1}^{n_E+1}  r_{1}^{n_C}(1-r_{1})^{n_E+1-n_C} \binom{n_E+1}{n_C} - \sum_{n_C = n_{C}^{-}}^{n_{E}} r_{1}^{n_C}(1-r_{1})^{n_E-n_C} \binom{n_E}{n_C} \nn \\
\ee \sum_{n_C = n_{C}^{-}+1}^{n_E}  r_{1}^{n_C}(1-r_{1})^{n_E-n_C} \bk{ (1-r_{1})\binom{n_E + 1}{n_C}- \binom{n_E}{n_C} } \nn \\
&+& r_{1}^{n_E+1} -  r_{1}^{n_{C}^{-}}(1-r_{1})^{n_{E}-n_{C}^{-}} \binom{n_E}{n_{C}^{-}} \nn \\
\ee \sum_{n_C = n_{C}^{-}+1}^{n_E}  r_{1}^{n_C}(1-r_{1})^{n_E-n_C} \bk{-r_{1} + \frac{n_{C}}{n_{E}-n_{C}+1}(1-r_{1})} \binom{n_E}{n_{C}} \nn \\
&+& r_{1}^{n_E+1} -  r_{1}^{n_{C}^{-}}(1-r_{1})^{n_{E}-n_{C}^{-}} \binom{n_E}{n_{C}^{-}} \label{ediff}
\,,}
where we used Pascal's identity $\dbinom{n-1}{k} + \dbinom{n-1}{k-1} = \dbinom{n}{k}$ and $\dbinom{n}{k-1}=\dfrac{k}{n+1-k}\dbinom{n}{k}$ in the last line.

Using the following result
\eqn{\sum_{n_C=n_{C}^{-}}^{n_{E}} \binom{n_{E}}{n_C} = \binom{n_{E}}{n_{C}^{-}} \, _2F_1(1,n_{C}^{-}-n_{E};n_{C}^{-}+1;-1) \,,}
where $_2F_1$ is the hypergeometric function, it can be shown after some algebra that Eq.~(\ref{ediff}) simply reduces to 
\eqn{\epsilon_{\mathrm{fail}}(n_E+1) - \epsilon_{\mathrm{fail}}(n_E) &\leq& -(1-r_{1})^{n_{E}-n_{C}^{-}+1} r_{1}^{n_{C}^{-}} \binom{n_E}{n_{C}^{-}}\,,}
which is always negative for any $n_{C}^{-}$. Moreover, Eve adding extra photons will always result in deleting the lowest term in the summation in Eq.~(\ref{e1ideal}) so that the failure probability monotonically decreases for all $n_{E}\geq n_{C}^{-} + n_{R}^{-} - 1$. Thus, the optimal value for Eve to maximise the failure probability is the single Fock state with photon number $n_E^{\mathrm{opt}} = n_{C}^{-} + n_{R}^{-} - 1$. Substitution into Eq.~(\ref{e1ideal}) then gives
\eqn{\epsilon_{\mathrm{fail}} &\leq& \sum_{n_C = n_{C}^{-}}^{n_E^{\mathrm{opt}}}  r_{1}^{n_C}(1-r_{1})^{n_E^{\mathrm{opt}}-n_C} \binom{n_E^{\mathrm{opt}}}{n_C} \nn\\
&\leq& \exp\bk{-2\frac{(n_{C}^{-}-r_{1} n_E^{\mathrm{opt}})^2}{n_E^{\mathrm{opt}}}} \,,} 
where the last line is given by Hoeffding's inequality which states that for a binomial distribution $\mathcal{B}(r_{1},n_{E},k)$ with $n_{C}^{-}\geq n_{E} r_{1}$, one gets
\eqn{\sum_{k=n_{C}^{-}}^{n_{E}} \mathcal{B}(r_{1},n_{E},k) \leq \exp \bk{-2\frac{(n_{C}^{-}-r_{1} n_{E})^2}{n_{E}}}\,.}

Finally, the probability that any one of the $m$ rounds fails is the complement that all of them pass thus 
\eqn{\epsilon_{\mathrm{fail,m}} = 1-(1-\epsilon_{\mathrm{fail}})^m \leq 1-(1-m \epsilon_{\mathrm{fail}}) = m \epsilon_{\mathrm{fail}} \,,} 
which is precisely the result stated Eq.~(\ref{esec}), thereby completing the proof. 

{\bf Completeness:}
Substituting in the number state expansion for a coherent state $\ket{\alpha}$ and calculating the probability for the certification test to pass via Eq.~(\ref{certmeas}) gives the desired result expressed in Eq.~(\ref{completeA}).

\end{proof}

\section{Modelling Detectors} \label{detectors}
Here, we remove the idealised assumptions from the previous section and present a detailed detector model.

\subsection{Finite range of photodetectors}
As a first idealisation, we shall remove the assumption of infinite dynamic range for the photodiodes. In fact, the detectors only respond linearly above and below certain photon numbers thresholds, namely $n_{\mathrm{min}}$ and $n_{\mathrm{max}}$. In reality, as the detectors enter this nonlinear regime, there will still be quantum randomness in their outcome statistics, but we take the worst case view and assume that all states with overly large or small photon numbers will be mapped with certainty to ``end bins'', thereby yielding no such randomness. Thus, instead of a sum over all photon number states, we model a photodetection with $L:= n_{\mathrm{max}}-n_{\mathrm{min}}+1$ measurement operators given by
\eqn{\hat{N}(n_{\mathrm{min}}) \ee \sum_{n=0}^{n_{\mathrm{min}}} \ket{n}\bra{n}\nn \,,\\
 \hat{N}(n) \ee \ket{n}\bra{n}\,, \hspace{2mm} \forall  \hspace{2mm}n_{\mathrm{min}}<n< n_{\mathrm{max}} \nn \,,\\
\hat{N}(n_{\mathrm{max}}) \ee \sum_{n=n_{\mathrm{max}}}^{\infty} \ket{n}\bra{n} \label{finiterange} \,.}

This can make quite a difference to the output randomness since if Eve either inputs a sufficiently small or large number of photons, she can be sure that the lower or upper outcome will occur on detectors A and B, leading to a difference outcome of 0 with certainty. This can be seen directly by calculating the difference measurement POVM elements using finite range photodetectors as an operator in Eve's input Hilbert space as before to find
\eqn{\hat{X}_{\mathrm{fin}}(x) \ee \left \{ \begin{array}{cc} 
\sum\limits_{n_A = n_{\mathrm{min}} + |x|}^{n_{\mathrm{max}}} \hat{M}(n_A,n_A-|x|)\,,&  \hspace{1mm} x\geq0\\
\sum\limits_{n_A = n_{\mathrm{min}} + |x|}^{n_{\mathrm{max}}} \hat{M}(n_A-|x|,n_A)\,, &  \hspace{1mm} x<0
\end{array} \right.
\label{Xfin}\,,}
where
\eqn{\hat{M}(n_A,n_B) =  \bra{0}\hat{U}^{\dag}_{BS_{0}}\hat{N}(n_A)\otimes\hat{N}(n_B) \hat{U}_{BS_{0}}\ket{0}\label{Mreal}\,.}

For states with an appropriate photon number support, a difference measurement made using finite range photodetectors will be virtually indistinguishable from the ideal difference measurement in Eq.~(\ref{x}). Specifically, if a number state $\ket{n}$ is input to a difference measurement with two detectors A and B that have linearity ranges $[n_{\mathrm{min}},n_{\mathrm{max}}]$ such that $n_{\mathrm{min}}\ll n/2 \ll n_{\mathrm{max}}$, then the probability that either detector will register a number of photons outside its linear range will be given by the tails of a binomial distribution. 
It can then be checked whether this probability is smaller than the other failure probabilities in the protocol (typical realistic values will render it far smaller, i.e. $\propto 1\times 10^{-30000}$). Alternatively, one can also directly empirically verify the linear response range $[n_{\mathrm{min}}^D,n_{\mathrm{max}}^D]$ of a difference measurement by inputting a known photonic laser source and observing that the difference variance indeed grows linearly when the laser's optical power is increased.

This finite range of the photodetection also applies to the certification measurement in mode C using a finite range detector with linear range $[n^C_{\mathrm{min}},n^C_{\mathrm{max}}]$ and $L_C = n^C_{\mathrm{max}}-n^C_{\mathrm{min}}+1$ possible outcomes. We have
\eqn{\hat{N}_{C,\mathrm{fin}}(n^C_{\mathrm{min}}) \ee \sum_{n_C=0}^{n^C_{\mathrm{min}}} \hat{N}_C(n_C)\nn\,, \\
\hat{N}_{C,\mathrm{fin}}(n_C) \ee \hat{N}_C(n_C)\,, \hspace{2mm} \forall  \hspace{2mm}n^C_{\mathrm{min}}<n_C< n^C_{\mathrm{max}}\nn \,, \\
\hat{N}_{C,\mathrm{fin}}(n^C_{\mathrm{max}}) \ee 
\sum_{n_C=n^C_{\mathrm{max}}}^{\infty} \hat{N}_C(n_C)\,,} 
where $\hat{N}_C(n_C)$ is given in Eq.~(\ref{Ncert}).

Finally, we can also write the failure operator associated with this certification measurement. It will be similar to the ideal case in Eq.~(\ref{certmeas}) except for the end bins. The failure of the protocol occurs when the test is passed and there are either too many (more than $n_{R}^{+}$) or too few (less than $n_{R}^{-}$) photons incident onto the difference measurement. We obtain the following failure operator
\eqn{\hat{F}(n^C_{\mathrm{min}},n_{R}^{-},n_{R}^{+}) \ee \sum_{n_C =0}^{n^C_{\mathrm{min}}} \left (\sum_{n_R=0}^{n_{R}^{-}} \frac{r_{1}^{n_C}(1-r_{1})^{n_R}(n_C+n_R)!}{n_C!n_R!} \ket{n_C+n_R}\bra{n_C+n_R}_{E} \right .\nn\\
 &+& \left. \sum_{n_R=n_{R}^{+} +1}^{\infty} \frac{r_{1}^{n_C}(1-r_{1})^{n_R}(n_C+n_R)!}{n_C!n_R!} \ket{n_C+n_R}\bra{n_C+n_R}_{E} \right ) \nn \,,\\
\hat{F}(n^C_{\mathrm{max}},n_{R}^{-},n_{R}^{+})  \ee \sum_{n_C =n^C_{\mathrm{max}}}^{\infty} \left (\sum_{n_R=0}^{n_{R}^{-}} \frac{r_{1}^{n_C}(1-r_{1})^{n_R}(n_C+n_R)!}{n_C!n_R!} \ket{n_C+n_R}\bra{n_C+n_R}_E \right . \nn\\
&+& \left. \sum_{n_R=n_{R}^{+} +1}^\infty \frac{r_{1}^{n_C}(1-r_{1})^{n_R}(n_C+n_R)!}{n_C!n_R!} \ket{n_C+n_R}\bra{n_C+n_R}_E\right )\nn \,,\\
 \hat{F}(n_C,n_{R}^{-},n_{R}^{+}) \ee \sum_{n_R=0}^{n_{R}^{-}} \frac{r_{1}^{n_C}(1-r_{1})^{n_R}(n_C+n_R)!}{n_C!n_R!} \ket{n_C+n_R}\bra{n_C+n_R}_E\nn\\
  &+& \sum_{n_R=n_{R}^{+} +1}^{\infty} \frac{r_{1}^{n_C}(1-r_{1})^{n_R}(n_C+n_R)!}{n_C!n_R!} \ket{n_C+n_R}\bra{n_C+n_R}_E\nn \,,\\
&\forall& \hspace{2mm} n_{\mathrm{min}}^C < n_C < n_{\mathrm{max}}^C \label{certmeasfinite} \,.}

\nw{Parenthetically, we note that finite-range considerations expose a problem with the proposed solution to saturation attacks found in \cite{furrer2014reverse} within the context of continuous-variable QKD. There, the idea is to tap off a small amount of the incoming light and measure it via a dual-homodyne (heterodyne) detection, aborting the protocol if a sufficiently large value of the heterodyne measurement is observed. While this solves the problem in the limit of perfect, infinite-range detectors, for any realistic finite-range detectors, this procedure itself is vulnerable to a saturation attack. To see this, consider an individual homodyne detection of one of the two field quadratures: the incoming signal is mixed with a local oscillator and the difference between the two detectors' signals is taken. However, a sufficiently bright input signal would saturate each individual detector such that it outputs its maximum value, meaning that the difference measurement would result in a (typically small) constant value. Thus, in contrast to our certification measurement based upon a single detector, there is no guarantee that a bright input would result in a large measurement outcome, and therefore applying a threshold check to a heterodyne detection offers no protection against high energy attacks. This again highlights the importance of rigorously modelling the trusted devices in a cryptographic setup, as even small imperfections can completely alter the security of the protocol.}

\subsection{Voltage response and temporal behaviour}
The next step in our modelling is to take into account the fact that the detector response is not completely flat over the time window that makes up one round of the protocol. Instead, the voltage response decays exponentially in time. However, using careful spectral filtering, one can enforce an effectively flat temporal distribution for incoming photons. Considering this, we show that we can model the voltage response with a single average conversion factor $\alpha$.

In general, the detector response of a photodiode can be regarded as analogous to a RC circuit where the voltage at time $T$ is given by
\eqn{V(T) \ee  \frac{1}{C} \int_{0}^{\infty} e^{-\tau/RC} I(T-\tau) \,d\tau \,,} 
where $I(T-\tau)$ is the current generated by the absorbed photons. However, one cannot take the above equation too literally since a genuinely continuous time dependence would correspond to a detector with infinite temporal resolution. Instead, we model a voltage detector as having K finite time intervals $\delta_t = T/K$ over which the response is flat (i.e. the detector cannot resolve temporal differences smaller than $\delta_t$). The entire detection over the time window $T$ can then be regarded as post-processing of the $K$ outcomes arising from each of the detection intervals $\delta_t$. This resulting POVM has elements of the form
\eqn{\hat{M}({\bf n}) = \hat{N}(n_1)\otimes\hat{N}(n_2)...\otimes\hat{N}(n_K)\label{timeM} \,,}
where ${\bf n} = [n_1,n_2,...n_K]$.
The voltage response to a photon arriving at the $k^{\mathrm{th}}$ interval is given by a conversion factor 
\eqn{\alpha_k: = \beta e^{-(K-k)BW \delta_t} \label{conversion_time}\,,}
where $\beta$ is a constant. The voltage POVM is thus expressed as
\eqn{\hat{V}(v) = \sum_{{\bf n}} c_{n,k}(v)\hat{M}(\bm{n})\,,}
with
\eqn{c_{n,k}(v) \ee \delta(v - {\bf n}\bm{ \alpha}^T )\,,}
where $\bm {\alpha}^T = [\alpha_1,..,\alpha_k]^T$ and the sum is over all $L^K$ possible values for ${\bf n}$. 

In principle, this temporal detector response could open loopholes for Eve to exploit. For example, if she were able to generate extremely short time pulses, Eve could saturate individual detectors which would then be heavily damped in time (due to the exponential term in Eq.~(\ref{conversion_time})), resulting in a certification voltage that would appear acceptable even though there would be no randomness in this case. However, these temporal attacks can be circumvented via an appropriate choice of spectral filtering in the detection process. For transform-limited pulses, a sufficiently narrow spectral filter {\it enforces} an effectively flat temporal distribution for the detected photons. Since the source in our experiment is extremely narrowband (single frequency CW laser), we can afford to use a correspondingly narrow filter without altering the detection rates in our actual implementation. Note that a pulsed system which cannot afford to be similarly filtered without reducing the resulting count rates would require a careful analysis of the effects of Eve's temporal modulation of the source on the output statistics. This highlights the importance of considering \textit{all} relevant physical degrees of freedom in certified randomness generation. 

Considering our implementation, the voltage response of a detector to a photon arrival is given by a time averaged conversion factor 
\eqn{\alpha: = \frac{hc BW \eta G}{\lambda} \label{conversion_time_flat}\,,}
where $h$ is Planck's constant, $c$ is the speed of light, $BW$ is the detector's bandwidth, $\eta$ is its responsitivity (in $\si{\ampere\per\watt}$) at the wavelength $\lambda$ considered and $G$ is the transimpedence gain.

\subsection{Electronic Noise}
So far, all measurements have been described without the presence of detector noise. As outlined in the main text, our detector's noise $\lambda$ is well modelled as being Gaussian with variance $\sigma^2$. We want to write down the POVM describing a voltage measurement over an appropriate basis as parameterised by its outcome. Given that the noisy measurement is still phase insensitive, the POVM elements can be written diagonally in the Fock basis as
\eqn{\hat{V}^{\sigma}(v) = \sum_{n=n_{\mathrm{min}}}^{n_{\mathrm{max}}} \frac{e^{-(v-\alpha n)^2/(2
\sigma^2)}}{\sqrt{2\pi}\sigma} \hat{N}(n)\label{voltt} \,.}

Consider the randomness generation measurement. Since the detector noise terms are taken to be independent from one another, we can equivalently combine them into a single overall noise variable $\lambda_D$ with variance $\sigma_D^2 =  \sigma_A^2 +\sigma_B^2$ (this joint variable is what was determined in practice during device calibration) that acts to smear out the ideal difference measurement to obtain\footnote{For detectors with the same conversion factor $\alpha$, a particular outcome at the detectors A and B would lead to a difference value $d = n_A- n_B + \lambda_A - \lambda_B = x + \lambda_D$ where we have combined the independent noise variables.}
\eqn{\hat{V}_D^{\sigma_D}(v_D) = \sum_{x=-(L-1)}^{L-1} \frac{e^{-(v_{D}-\alpha_{D}x)^2/(2\sigma_D^2)}}{\sqrt{2\pi}\sigma_D} \hat{X}_{\mathrm{fin}}(x)\label{volts}\,,}
with $\hat{X}_{\mathrm{fin}}(x)$ given by Eq.~(\ref{Xfin}) but effectively by Eq.~(\ref{x}) for the photon ranges we will certify.

In addition, the certification measurement's POVM accounting for the Gaussian noise characterised by variance $\sigma_C^2$ is given by
\eqn{\hat{V}_C^{\sigma_C}(v_C) = \sum_{n=n^C_{\mathrm{min}}}^{n^C_{\mathrm{max}}} \frac{e^{-(v_C-\alpha_C n_C)^2/(2
\sigma_C^2)}}{\sqrt{2\pi}\sigma_C} \hat{N}_{C,\mathrm{fin}}(n_C) \,.}

Finally, for the failure operator associated with the certification measurement with Gaussian electronic noise, we have the following 
\eqn{\hat{V}_F^{\sigma_C}(v_C,n_{R}^{-},n_{R}^{+}) = \sum_{n_C = n^C_{\mathrm{min}}}^{n^C_{\mathrm{max}}} \frac{e^{-(v_C-\alpha_C n_C)^2/(2\sigma_C^2)}}{\sqrt{2\pi}\sigma_C} \hat{F}(n_C,n_{R}^{-},n_{R}^{+}) \,,}
where $\alpha_C$ is the voltage conversion factor for the photodetector C and $\sigma_C$ is the standard deviation of its associated electronic noise.
 
For the security analysis later, we will often be interested in the measurement operators from Eve's perspective who always knows the relevant value of $\lambda$. This leads to a voltage POVM given by
\eqn{\hat{V}(v) = \hat{N}\bk{\frac{v - \lambda}{\alpha}}\,,}
a difference measurement
\eqn{\hat{V}_D(v_{D}) =  \hat{X}_{\mathrm{fin}}\bk{\frac{v_{D}-\lambda_D}{\alpha_{D}}}\,,}
a certification measurement
\eqn{\hat{V}_C(v_C) = \hat{N}_{C,\mathrm{fin}} \bk{\frac{n_C-\lambda_C}{\alpha_C}}\,,}
and a failure operator associated with certification voltage measurement
\eqn{\hat{V}_F(v_C,n_{R}^{-},n_{R}^{+}) =  \hat{F}\bk{\frac{v_C-\lambda_C}{\alpha_C},n_{R}^{-},n_{R}^{+}} \label{failureEve} \,.}

\subsection{\nw{Finite resolution and range of analogue-to-digital converter}}
In the previous section, we modelled the detectors as having a finite range but otherwise being perfectly photon-number resolving and convolved with a classical noise variable subsequently given to the eavesdropper. In fact, the randomness generation measurement has a finite resolution which corresponds to an extra coarse graining. Specifically, the analogue-to-digital converter (ADC) which processes the voltage signal can only record a certain set range of voltages $[V_{\mathrm{min}},V_{\mathrm{max}}]$, with all voltages greater or smaller than this amount registered as results in the ``end bin''. Furthermore, within the range $[V_{\mathrm{min}},V_{\mathrm{max}}]$, voltages are only recorded with a finite resolution. Therefore, whilst an ideal voltage measurement might have unbounded and continuous values, a real detector in combination with an ADC with finite bits of resolution $\Delta_{\mathrm{ADC}}$ outputs $J=2^{\Delta_{\mathrm{ADC}}}$ outcomes with corresponding POVM elements $\{ \hat{V}^{\sigma,\Delta_{\mathrm{ADC}}}(j)\}_{j}$ for the measured $j^\mathrm{th}$ bin expressed as
\eqn{\hat{V}^{\sigma,\Delta_{\mathrm{ADC}}}(j) = \int_{I_j} \hat{V}^{\sigma}(v) \, dv \label{voltfinal}\,,}
where the integration regions are given by 
\eqn{I_{\left \lfloor -(J-1)/2\right \rfloor} \ee [-\infty, V_{\mathrm{min}}+\delta V[,\quad I_{\left \lfloor-(J-1)/2 +1\right \rfloor} = [V_{\mathrm{min}} + \delta V,V_{\mathrm{min}} + 2\delta V[,\quad\dots \nn \\
&&\dots,\quad I_0 = [-\delta V/2, \delta V/2[,\quad\dots,\quad I_{\left \lceil(J-1)/2\right \rceil} = [V_{\mathrm{min}}+(J-1)\delta V,\infty[ \label{intervals}\,,}
and $\delta V = \frac{V_{\mathrm{max}}-V_{\mathrm{min}}}{2^{\Delta_{\mathrm{ADC}}}}$ is the effective voltage resolution induced by $\Delta_{\mathrm{ADC}}$. \nw{Note that $\lfloor\cdot\rfloor$ and $\lceil\cdot\rceil$ are the floor and ceiling functions, respectively.}

As a result, the coarse grained noisy difference measurement operators are given by $\{ \hat{V}_D^{\sigma_D,\Delta_{\mathrm{ADC}}}(j)  \}_{j}$ for which
\eqn{\hat{V}_D^{\sigma_D,\Delta_{\mathrm{ADC}}}(j) =  \int_{I_j^D} \hat{V}_D^{\sigma_D}(v_{D}) \, dv_{D} \label{measfinal} \,.}

The corresponding difference measurement from Eve's perspective (i.e. given the relevant $\lambda$) would be
\eqn{\hat{V}_D^{\Delta_{\mathrm{ADC}}}(j) \ee  \int_{I_j^D-\lambda_D} \hat{V}_D(v_{D}) \, dv_{D} \nn\\
\ee \sum_{x\in\mathcal{X}} \hat{X}_{\mathrm{fin}}(x) \label{measEve} \,,}
where
\eqn{\mathcal{X} \ee \{x : \alpha_{D} x + \lambda_{D} \in I_{j}^D \}\,.}

The certification voltage measurement is recorded by an ADC with the same resolution and consequently it is still a $J$-outcome measurement but over an ADC range $[V_{\mathrm{min}}^C,V^C_{\mathrm{max}}]$ and a corresponding voltage resolution $\delta V_C = \frac{V^C_{\mathrm{max}}-V^C_{\mathrm{min}}}{2^{\Delta_{\mathrm{ADC}}}}$. This leads to intervals $I_i^C$ which are defined as per Eq.~(\ref{intervals}) and coarse-grained certification measurements elements
\eqn{\hat{V}_C^{\sigma_C,\Delta_{\mathrm{ADC}}}(i) = \int_{I_i^C} \hat{V}_C^{\sigma_C}(v_C) \, dv_C\label{certfinal}\,.}

Moreover, the associated failure operator is 
\eqn{\hat{V}_F^{\sigma_C,\Delta_{\mathrm{ADC}}}(i,n_{R}^{-},n_{R}^{+}) = \int_{I_i^C} \hat{V}_F^{\sigma_C}(v_C,n_R^-,n_R^+) \, dv_C \,.}

For a fixed value of the noise variable $\lambda_C$, we have the following failure operator from Eve's perspective
\eqn{\hat{V}_F^{\Delta_{\mathrm{ADC}}}(i,n_{R}^{-},n_{R}^{+}) \ee \int_{I^C_{i} - \lambda_C} \hat{V}_F(v_C,n_{R}^{-},n_{R}^{+}) \, dv_C \nn\\ 
\ee \sum_{n_C\in \mathcal{C}} \hat{F}(n_C,n_{R}^{-},n_{R}^{+})\label{K_no_lambda} \,,} 
where 
\eqn{\mathcal{C} = \{ n_C: \alpha_C n_C + \lambda_C \in I_{i}^C \}\,.}

In general, one must be mindful of the interplay between the conversion from photon number to voltage and the final voltage resolution. Indeed, if the signal were to experience strong attenuation (very small $\alpha$), then the voltage distribution would start to become small with respect to the fixed voltage resolution and the entropy would decrease. In our implementation, we carefully kept track of the coarse graining, thus avoiding such issue. 

Before we proceed further, we show in Fig.~\ref{fig:detector_model} a schematic drawing summarising our detector's model. The POVMs present in the figure are those specified in this appendix.
\begin{figure}[h!]
\boxed{\includegraphics[scale=1]{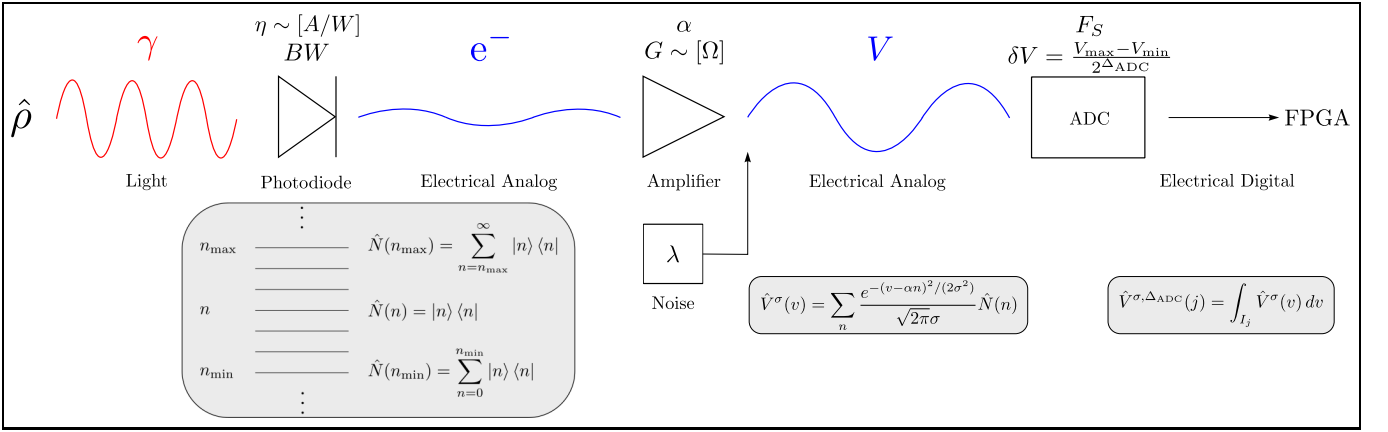}}
\caption{Detector model. Photons from a photonic state $\hat{\rho}$ impinge onto a photodiode whose linear range and equivalent $L$ photon projectors are given in Eq.~(\ref{finiterange}). The photodiode's voltage response is given by the conversion factor $\alpha$ expressed in Eq.~(\ref{conversion_time}) in general and Eq.~(\ref{conversion_time_flat}) in our case. This factor incorporates the photodiode's bandwidth $BW$, its responsitivity $\eta$ (in $\si{\ampere\per\watt}$) and the transimpedence gain $G$. Noise characterised by a Gaussian random variable $\lambda$ is then added onto the voltage, leading to the voltage POVM in Eq.~(\ref{voltt}). Finally, the voltage is discretised by an ADC with effective resolution $\delta V$ and at a sampling rate $F_{S}$, yielding the POVM associated with the measurement of the $j^{\mathrm{th}}$ voltage bin expressed in Eq.~(\ref{voltfinal}). \nw{Light has been effectively converted from photons to a digital electrical signal which one can subsequently feed to an FPGA.}}
\label{fig:detector_model}
\end{figure}

\section{Proof of the Main Theorem} \label{proof_main_theorem}
In this Appendix, we provide the full security proof for the more realistic QRG protocol carried out in the experiment. As per the idealised protocol, the proof proceeds in two steps. First, we calculate the worst-case min-entropy for a certain class of states, namely those with a limited support over Fock states. Secondly, we calculate the failure probability of the protocol which is the maximum probability that a state not in that class could have passed the certification test. We rewrite theorem \ref{rgend1} given in the main text and proceed with our proof. 

\begin{Theorem} \label{rgend2} 
An optical setup consisting of 
\begin{itemize}
\item Two trusted vacuum modes
\item Two beam splitters of reflectivity $r_{0}=\frac{1}{2}$ and $r_{1}$
\item Two noisy photodetectors used to make a difference measurement as described in Eq.~(\ref{measfinal})
\item A third noisy photodetector used to make a certification measurement as described in Eq.~(\ref{certfinal}) which passes the test $\mathcal{P}$ if $i$ falls in a chosen range $[i_-,i_+]$
\end{itemize}

can be used as a certified (m,$\kappa$,$\epsilon_{\mathrm{fail,m}}$,$\epsilon_c$)-randomness generation protocol as per Definition \ref{QRGdef} without making any assumptions about the photonic source with
\eqn{\kappa &\geq& - m\log_2 \left ( \sum_{x\in \mathcal{X}}2^{-n_{R}^{-}}\binom{n_{R}^{-}}{\lfloor \frac{n_{R}^{-}+x}{2} \rfloor} \right ) \label{hminthm}\,,}
where
\eqn{\mathcal{X} \in \mathbb{N} \cap \left [ -\left \lfloor \frac{\delta V}{2\alpha_{D}} \right \rfloor, \left \lfloor \frac{\delta V}{2 \alpha_{D}} \right \rfloor  \right ] \label{support}\,,}
with $\delta V = \frac{V_{\mathrm{max}}-V_{\mathrm{min}}}{2^{\Delta_{\mathrm{ADC}}}}$,

\eqn{\label{mround} \epsilon_{\mathrm{fail,m}} &\leq&m \epsilon_{\mathrm{fail}} \,,}
where
\eqn{\epsilon_{\mathrm{fail}} = \max \{ \epsilon_-,\epsilon_+\} + \epsilon_{\lambda_{C}} \label{epsilontheorem}\,,}
with 
\eqn{ \epsilon_- \ee \exp\bk{-2\frac{\left(\frac{v_{i_{-}}^- - \tilde{\lambda}}{\alpha_C}-r_{1} \left(\frac{v_{i_{-}}^- - \tilde{\lambda}}{\alpha_C} + n_{R}^{-} - 1\right)\right)^2}{\frac{v_{i_{-}}^- - \tilde{\lambda}}{\alpha_C} + n_{R}^{-} - 1}} \,, \nn \\
\epsilon_+ \ee \exp\bk{-2\frac{\left(n_{R}^{+} -(1-r_{1}) \left(\frac{v_{i_{+}}^+ - \tilde{\lambda}}{\alpha_C} + n_{R}^{+} + 1\right)\right)^2}{\frac{v_{i_{+}}^+ - \tilde{\lambda}}{\alpha_C} + n_{R}^{+} + 1}} \,, \nn\\ 
\epsilon_{\lambda_{C}} \ee 1 - \mathrm{erf}\left(\frac{\tilde{\lambda}}{\sqrt{2} \sigma_C }\right) \label{esecreal} \,,}
provided $n_R^+$ is set to the saturating photon number of the difference measurement. 

Moreover, \eqn{\epsilon_c = 1 - \mathrm{tr} \left\{\sum_{i=i_-}^{i_+}\ket{\alpha}\bra{\alpha} \hat{V}_C^{\sigma_C,\Delta_{\mathrm{ADC}}}(i) \right\} \label{complete} \,,} 
using a coherent state $\ket{\alpha}$ as an input.
\end{Theorem}

\begin{proof}
{\bf Security:}
Consider the task of guessing the difference measurement from the perspective of Eve who knows $\lambda_D$ on a shot-by-shot basis, which is given by Eq.~(\ref{measEve}). First, this measurement satisfies the conditions of Lemma \ref{fockopt} and so Eve's optimal state is a number state. Her strategy will be to add $\lambda_D$ to the most likely value of the noiseless difference measurement which, as shown in Appendix \ref{ideal}, is 0 or 1 depending upon whether Eve inputs an odd or even number of photons. Therefore, Eve's best guess will be the voltage bin $I_j^D$ with $j=\left [ \frac{\lambda_{D}}{\delta V}\right ] $ or $j=\left [ \frac{(1+\lambda_{D})}{\delta V}\right ]$, where $[.]$ is the nearest integer rounding function. The guessing probability is given by the sum of all the probabilities associated with the outcomes $\hat{X}(x)$ for which Eve's guess would remain true. This can be expressed as the following set
\eqn{\mathcal{X} \ee \{x \in [-(L-1),L-1] :\hspace{2mm} \alpha_{D} x + \lambda_{D} \in I_{j}^D \}\,.}

For states restricted to the range $[n_{R}^{-},n_{R}^{+}]$, the guessing probability corresponds to
\eqn{p_{\mathrm{guess}} \ee  \max_{n\in [n_{R}^{-},n_{R}^{+}]}  \bra{n}\sum_{x \in \mathcal{X}} \hat{X}(x) \ket{n} \,,}
where again the sum only includes even (odd) values of $x$ when $n$ is even (odd).

From the expressions above, the interplay between the voltage conversion factor $\alpha_D$ and the voltage resolution $\delta V$ becomes clear. The number of difference measurement elements that will be mapped to a given voltage bin is given by $\left\lceil \frac{\delta V}{\alpha_{D}}\right\rceil$, such that as $\alpha_D$ becomes smaller, this number grows and Eve's guessing probability will increase. Since we will only consider number states within the linear regime of the difference measurement (i.e. $n_{R}^{+} = n_{\mathrm{max}}$), we can safely assert that $\bra{n}\hat{X}(x) \ket{n} = 2^{-n}\binom{n}{\lfloor \frac{n+x}{2} \rfloor}$ is a binomial distribution. Thus, the largest guessing probability for a given $n$ will occur when $\lambda_D$ is such that the $\left\lceil \frac{\delta V}{\alpha_{D}}\right\rceil$ bins are centered evenly around the origin, i.e. the middle portion of the binomial distribution. Moreover, we know from Section \ref{ideal} that the guessing probability will decrease monotonically with the photon number. This yields
\eqn{p_{\mathrm{guess}} \leq  \sum_{x\in \mathcal{X}}2^{-n_R^-}\binom{n_{R}^{-}}{\lfloor \frac{n_{R}^{-}+x}{2} \rfloor}\,,}
which is exactly Eq.~(\ref{hminthm}). While this expression can be directly evaluated numerically, for large $n_{R}^{-}$ (recall here that $n_{R}^{-} >10^{5}$), one can use the Gaussian distribution as an excellent approximation for the binomial distribution and evaluate the sum as an integral to get the following compact expression
\eqn{p_{\mathrm{guess}} \leq \frac{1}{2} \left(\erf\left(\frac{\frac{\delta V}{2 \alpha_{D} }}{ \sqrt{\frac{n_R^-}{4}}}\right)-\erf\left(\frac{-\frac{\delta V}{2 \alpha_{D} }-1}{\sqrt{\frac{n_R^-}{4}}}\right)\right)\,.}

The failure probability for the protocol is given by the probability of passing the test even though a state with too few, or too many, photons is incident onto the difference measurement in mode R. We can express the probability of Eve successfully cheating in a single round as
\eqn{\epsilon_{\mathrm{fail}} \ee \max_{\hat{\rho}_E} \mathrm{Pr} \left[i^{-}\leq i \leq i^{+} \wedge n_R\notin [n_{R}^{-},n_{R}^{+}]\right] \nn\\
\ee \max_{\hat{\rho}_E} \mathrm{tr} \left \{ \hat{\rho}_E \sum_{i= i^{-}}^{i^{+}} \hat{V}_F^{\sigma_C,\Delta_{\mathrm{ADC}}}(i,n_{R}^{-},n_{R}^{+}) \right \} \nn\\
\ee \max_{n_E} \mathrm{tr} \left \{ \ket{n_E}\bra{n_E} \sum_{i= i^{-}}^{i^{+}} \hat{V}_F^{\sigma_C,\Delta_{\mathrm{ADC}}}(i,n_{R}^{-},n_{R}^{+})\right \}\,,\label{efail1}}
where in the last we line we used the fact that $\hat{V}_F$ satisfies the conditions of Lemma \ref{fockopt}, implying that Eve's optimal input state will be a number state. 

To begin with, let us consider this probability given a particular value for $\lambda_C$, the detector's noise variable. Then, from Eve's perspective, this electronic noise $\lambda_{C}$ is effectively removed as expressed in Eq.~(\ref{K_no_lambda}) and we have
\eqn{\epsilon_{\mathrm{fail}} \ee \max_{n_E} \mathrm{tr} \left \{  \ket{n_E}\bra{n_E}\sum_{n_C = n_C^-}^{n_C^+}\hat{F}(n_C,n_{R}^{-},n_{R}^{+}) \right \} \nn\\
\ee \max_{n_E} \mathrm{tr} \left \{ \ket{n_E}\bra{n_E} \sum_{n_C = n_C^-}^{n_C^+} \left ( \sum_{n_R=0}^{n_{R}^{-}}\mathcal{B}(r_{1},n_C+n_R,n_C)  \ket{n_C+n_R}\bra{n_C+n_R}_{E} \right. \right .\nn\\
  &+& \left. \left . \sum_{n_R=n_{R}^{+} +1}^{\infty} \mathcal{B}(r_{1},n_C+n_R,n_C) \ket{n_C+n_R}\bra{n_C+n_R}_{E} \right ) \right \} \nn\\
\ee  \max_{n_E} \left \{ \sum_{n_C=\max \{ n_C^-,n_E-(n_{R}^{-}-1)\}}^{\min\{n_C^+,n_E\}} \mathcal{B}(r_{1},n_E,n_C) +\sum_{n_R = \max\{n_{R}^{+}, n_E-(n_C^+ +1) \}}^{n_E} \mathcal{B}(1-r_{1},n_E,n_R) \right \}\label{e1final} \,,}
where $n_C^- = \min_{n_C} \{ n_C: \alpha_Cn_C + \lambda_C \in I_{[i^{-},i^{+}]}^C\}$ and $n_C^+ = \max_{n_C} \{ n_C: \alpha_Cn_C + \lambda_C \in I_{[i^{-},i^{+}]}^C\}$ with $I_{[i^{-},i{+}]}^C$ being the entire voltage range for which the test $\mathcal{P}$ is passed. 

Let $v_{i}^{\pm}=\delta V(i\pm \frac{1}{2})$ be the smallest and largest voltages corresponding to bin $i$. Therefore, the minimum (maximum) voltage consistent with passing the test is $v_{i_{-}}^{-}$ ($v_{i_{+}}^{+}$). The corresponding minimum and maximum photon numbers are
 \eqn{n_C^- \ee \frac{v_{i_{-}}^{-} - \lambda_C}{\alpha_C} \nn \,, \\
 n_C^+ \ee \frac{v_{i_{+}}^{+} - \lambda_C}{\alpha_C} \,.}
 
We can use our knowledge of the detector's noise distribution to turn these into worst case upper and lower bounds for $n_C^{+}$ and $n_C^{-}$, respectively. Recalling that $\lambda_C$ is Gaussian with variance $\sigma_C^2$, we can say that except with a probability
\eqn{\epsilon_{\lambda_{C}} = 1 -  \text{erf}\left(\frac{\tilde{\lambda}}{\sqrt{2} \sigma_C }\right) \label{erfbound}\,,} 
one has $|\lambda_C|<\tilde{\lambda}$. This gives
 \eqn{n_C^- &\geq& \frac{v_{i_{-}}^{-} - \tilde{\lambda}}{\alpha_C} \nn \,,\\
 n_C^+ &\leq& \frac{v_{i_{+}}^{+} + \tilde{\lambda}}{\alpha_C}\,.}

Next, the varying limits in the sums of Eq.~(\ref{e1final}) can be explained as follows. For the first sum, an unconditional lower limit is given by $n_C^-$. However, for sufficiently large inputs $n_E$, this requirement is superseded by the constraint that $n_R<n_{R}^{-}$, which in turn necessitates that $n_C\geq n_E-(n_{R}^{-}-1)$. The upper limit simply comes from the fact that if $n_E<n_C^+$, then the binomial distribution can only run up to $n_E$. For the second sum, we have an unconditional constraint $n_R>n_{R}^{+}$, however again for sufficiently large $n_E$, the requirement that $n_C<n_C^-$ implies that we must have $n_R>n_E-(n_C^++1)$. Notice that depending upon the bounds for $n_C^{+}$ and $n_C^{-}$, there are certain values of $n_E$ for which the first or second sums may vanish. This turns out to be the case here (i.e. for our values only one of the sums will be non-zero at a time).

The first sum in Eq.~(\ref{e1final}) will vanish whenever $n_E> n_{C}^{+}+n_{R}^{-}-1 \geq \frac{v_{i_{+}}^{+} - \tilde{\lambda}}{\alpha_C} + n_{R}^{-} - 1$ and the second when $n_E< n_{R}^{+}$. In summary, as long as 
\eqn{n_{R}^{+}&>&\frac{v_{i_{+}}^{+} - \tilde{\lambda}}{\alpha_C} + n_{R}^{-} - 1 \nn \\
\Rightarrow \tilde{\lambda} &\leq& v^+_{i^+} - \alpha_C\bk{n_R^+ - n_R^- + 1}\label{1meas}\,,} 
it implies that there are no values of $n_E$ for which both sums will be simultaneously nonzero. In our case, this condition evaluates to 
\eqn{|\tilde{\lambda}| \leq 1.155\,.}

We will always be making a much tighter probabilistic bound on $\tilde{\lambda}$ such that Eq.~(\ref{1meas}) is satisfied at all times. Substitution in Eq.~(\ref{erfbound}) indicates that this will be true except with probability $10^{-3769921}$, which is far below the other failure probabilities that we certify.

Except with probability $\epsilon_{\lambda_{C}}$, we can then write the single round failure probability as
\eqn{\epsilon'_{\mathrm{fail}} \ee \max \left \{ \max_{n_E}  \sum_{n_{C}=\max \{ n_C^-,n_E-(n_{R}^{-}-1)\}}^{\min\{n_C^+,n_E\}} \mathcal{B}(r_{1},n_E,n_C), \right. \nn\\
 &&\max_{n_E}  \left. \sum_{n_R = \max\{n_{R}^{+}, n_E-(n_C^+ +1) \}}^{n_E} \mathcal{B}(1-r_{1},n_E,n_R)  \right \} \,.}

Considering the first term, we have
\eqn{\max_{n_E}  \sum_{n_{C}=\max \{ n_C^-,n_E-(n_{R}^{-}-1)\}}^{\min\{n_C^+,n_E\}} \mathcal{B}(r_{1},n_E,n_C) \leq \max_{n_E}  \sum_{n_{C}=\max \{ n_C^-,n_E-(n_{R}^{-}-1)\}}^{n_E} \mathcal{B}(r_{1},n_E,n_C)\,.}

This expression is exactly the same as Eq.~(\ref{e1ideal}) for which we already know that $n_E^{\mathrm{opt}} = n_C^- + n_{R}^{-} - 1$. Therefore, we can apply Hoeffding's bound to the binomial cumulative distribution to obtain
\eqn{\max_{n_E}  \sum_{n_{C}= n_C^-}^{n_E} \mathcal{B}(r_{1},n_E,n_C) &\leq& \exp\bk{-2\frac{(n_C^--r_{1} (n_C^- + n_{R}^{-} - 1))^2}{n_C^- + n_{R}^{-} - 1}} \nn \\
&\leq& \exp\bk{-2\frac{\left(\frac{v_{i_{-}}^{-} - \tilde{\lambda}}{\alpha_C}-r_{1} \left(\frac{v_{i_{-}}^{-} - \tilde{\lambda}}{\alpha_C} + n_{R}^{-} - 1\right)\right)^2}{\frac{v_{i_{-}}^{-} - \tilde{\lambda}}{\alpha_C} + n_{R}^{-} - 1}} \,,}
provided there exists a $n_{R}^{-}$ such that $n_{R}^{-}> \frac{1-r_{1}}{r_{1}}\frac{v_{i_{-}}^{-} - \tilde{\lambda}}{\alpha_C}$.

The second term in the maximisation is again just the cumulative tail of a binomial distribution and via the same argument as in Eq.~(\ref{e1ideal}), we know that Eve should choose $n_E^{\mathrm{opt}} = n_{R}^{+}+n_C^+ +1$ to maximise this term, giving
\eqn{\sum_{n_R = n_{R}^{+}}^{n_E} \mathcal{B}(1-r_{1},n_E,n_R) &\leq& \exp\bk{-2\frac{(n_{R}^{+} -(1-r_{1}) (n_C^+ + n_{R}^{+} + 1))^2}{n_C^+ + n_{R}^{+} + 1}} \nn \\
&\leq& \exp\bk{-2\frac{\left(n_{R}^{+} -(1-r_{1}) \left(\frac{v_{i_{+}}^{+} - \tilde{\lambda}}{\alpha_C} + n_{R}^{+} + 1\right)\right)^2}{\frac{v_{i_{+}}^{+} - \tilde{\lambda}}{\alpha_C} + n_{R}^{+} + 1}} \,,}
provided there exists $n_{R}^{+}> \frac{1-r_{1}}{r_{1}}\frac{v_{i_{+}}^{+} - \tilde{\lambda}}{\alpha_C}$. 

Thus, the total failure probability for one round of the protocol is given by
\eqn{\epsilon_{\mathrm{fail}} \ee \epsilon_{\mathrm{fail}}' + \epsilon_{\lambda_{C}} \nn\\
\ee \max \left \{\exp\bk{-2\frac{\left(\frac{v_{i_{-}}^{-} - \tilde{\lambda}}{\alpha_C}-r_{1} \left(\frac{v_{i_{-}}^{-} - \tilde{\lambda}}{\alpha_C} + n_{R}^{-} - 1\right)\right)^2}{\frac{v_{i_{-}}^{-} - \tilde{\lambda}}{\alpha_C} + n_{R}^{-} - 1}},\right. \nn \\
&& \left. \exp\bk{-2\frac{\left(n_{R}^{+} -(1-r_{1}) \left(\frac{v_{i_{+}}^{+} - \tilde{\lambda}}{\alpha_C} + n_{R}^{+} + 1\right)\right)^2}{\frac{v_{i_{+}}^{+} - \tilde{\lambda}}{\alpha_C} + n_{R}^{+} + 1}} \right \} \nn \\ &+& 1 - \text{erf}\left(\frac{\tilde{\lambda}}{\sqrt{2} \sigma_C }\right) \,.}

which is exactly Eq.~(\ref{esecreal}), thereby completing the proof. 

{\bf Completeness:}
Lastly, the argument for completeness is the same as that in Appendix \ref{ideal}. 

\end{proof}

\section{\nw{Mathematical details}} \label{details}
\nw{Composable security for a protocol is frequently defined in terms of the probability of passing some test $p_{\mathrm{pass}}$, the distinguishability between the output of a real implementation conditioned on passing that test $\hat{\rho}_{\mathrm{pass}}$ and an ideal output of the protocol $\hat{\rho}_{\mathrm{ideal}}$. Since quantum state distinguishability is precisely captured by the trace distance $D(\hat{\rho},\hat{\sigma}) = ||\hat{\rho} - \hat{\sigma}||_1$, the security parameter of such a definition is typically written as $\epsilon_{\mathrm{fail}} := p_{\mathrm{pass}}D(\hat{\rho}_{\mathrm{pass}}, \hat{\rho}_{\mathrm{ideal}}) $. Above, we showed that the security parameter for this protocol is
\eqn{\epsilon_{\mathrm{fail}} \ee \max_{\hat{\rho}_E} \mathrm{tr} \left \{  \hat{\rho}_E\sum_{n_C = n_C^-}^{n_C^+}\hat{F}(n_C,n_{R}^{-},n_{R}^{+}) \right \}\,,\label{epsilonAppD}}
where the failure operators $\hat{F}(n_C,n_R^-,n_R^+)$ are defined in Eq.~(\ref{certmeasfinite}). 

This can be interpreted as as the joint probability that the test would be passed in mode C whilst a photon number outside the range $[n_{R}^-,n_R^+]$ was measured for $\hat{\rho}_R^{\mathrm{pass}}$ (the conditional state in mode R). For completeness, we will show here that $\epsilon_{\mathrm{fail}}$ can equivalently be seen as the probability of passing the test multiplied by the distinguishability between $\hat{\rho}_R^{\mathrm{pass}}$ and any state with support solely in the range $[n_{R}^-,n_R^+]$. Recall that without loss of generality, we can take Eve's input state $\hat{\rho}_{E}$ to be diagonal in the Fock basis. In this case, $\hat{\rho}_R^{\mathrm{pass}}$ will also be diagonal in the Fock basis and so will the closest state in the range $[n_{R}^-,n_R^+]$ which we will denote $\hat{\sigma}_{[n_{R}^{-},n_R^+]}$. For such diagonal states, the trace distance simplifies and it is straightforward to show that the distance $D(\hat{\rho}_R^{\mathrm{pass}},\hat{\sigma}_{[n_{R}^{-},n_R^+]})$ is just the probability of projecting $\hat{\rho}_R^{\mathrm{pass}}$ onto a Fock state that lies outside $[n_{R}^{-},n_R^+]$. In other words
\eqn{D(\hat{\rho}_R^{\mathrm{pass}},\hat{\sigma}_{[n_{R}^{-},n_R^+]}) = \mathrm{tr} \left\{\hat{\rho}_R^{\mathrm{pass}}\bk{\sum_{n_R = 0}^{n_R^-} \ket{n_R}\bra{n_R} + \sum_{n_R = n_R^+}^{\infty} \ket{n_R}\bra{n_R}} \right \}\,.}

However, this probability is precisely the same as the joint probability of observing too few or too many photons in mode R whilst passing the test, renormalised by the probability of passing the test. The joint probability is exactly what is given by the failure mode operators in Eq.~(\ref{certmeasfinite}) acting on Eve's input. Thus, we can write
\eqn{D(\hat{\rho}_R^{\mathrm{pass}},\hat{\sigma}_{[n_{R}^{-},n_R^+]}) = \frac{1}{p_{\mathrm{pass}}}\mathrm{tr} \left \{  \hat{\rho}_E\sum_{n_C = n_C^-}^{n_C^+}\hat{F}(n_C,n_{R}^{-},n_{R}^{+}) \right \}\,.\label{tdepsilon}}

Comparing Eq.~(\ref{tdepsilon}) with Eq.~(\ref{epsilonAppD}), we find
\eqn{\epsilon_{\mathrm{fail}} = p_{\mathrm{pass}} D(\hat{\rho}_R^{\mathrm{pass}},\hat{\sigma}_{[n_{R}^{-},n_R^+]})\,,}
which shows that our failure probability can also be interpreted as the product of $p_{\mathrm{pass}}$ and the distinguishing probability between the conditional output state and an ideal state (i.e. one that has support solely in the desired photon number range), as claimed in Appendix \ref{ideal}.}

\section{Source-device independent quantum random number expansion} \label{rne}
\nw{The certified SDI-QRG protocol either aborts or, except with some failure probability $\epsilon_{\mathrm{fail,m}}$, produces an output $X$ with a min-entropy $\hmin(X|E)\geq \kappa>0$ with respect to any third party, even one with complete control over the photonic source and access to all other environmental modes. Equivalently, this is the joint probability of simultaneously passing the certification test $\mathcal{P}$ and producing an output with less than a specified amount of min-entropy, expressed as
\eqn{p_{\mathrm{pass}} \mathrm{Pr}[\hmin(X|E)<\kappa] \leq \epsilon_{\mathrm{fail,m}}\,.}}

However, the final goal of a randomness expansion protocol is to utilise an initial random seed in order to generate a much longer bit string that is ``$\epsilon$-close'' (in some well chosen metric) to perfectly uniformly distributed and unpredictable with respect to any third party. This can be achieved via randomness extraction (also sometimes called privacy amplification), which is a judiciously chosen post-processing of the measurements. We would also like to be confident that a realistic implementation of the protocol will succeed with high probability. Without loss of generality, the output state $S$ of this post-processing can be written as a classical-quantum state
\eqn{\hat{\rho}_{SE} = \sum_{s} P_{S}(s) \ket{s}\bra{s}\otimes\hat{\rho}_E^{s} \label{S} \,,}
for which we have the following definition.

\begin{Definition} \label{definition}
A protocol that outputs a state of the form in Eq.~(\ref{S}) is
\begin{itemize}
\item \textbf{Security:} $\epsilon_l$-secure (or sound) if 
\nw{\eqn{\hs p_{\mathrm{pass}}  D(\hat{\rho}_{SE}, \hat{\tau}_{S}\otimes \hat{\sigma}_E) \leq \epsilon_l \label{sec} \,,}
where $p_{\mathrm{pass}}$ is the probability that the certification test $\mathcal{P}$ is passed, $D(\hat{\rho},\hat{\sigma}):=\half ||\hat{\rho} - \hat{\sigma}||_1$ is the trace distance and $\hat{\tau}_{S}$ is the uniform (i.e. maximally mixed) state over $S$.} This means that there is no device or procedure that can distinguish between the actual protocol and an ideal protocol with probability higher than $\epsilon_{\mathrm{s}}$.
\item \textbf{Completeness:} {\it $\epsilon_c$-complete} (or robust) if there exists an honest implementation such that $1-p_{\mathrm{pass}}\leq \epsilon_c$.
\end{itemize}
\end{Definition}

The properties of the trace norm ensure that randomness satisfying Definition \ref{definition} is composable, which is critical for cryptographic applications \cite{portmann2014cryptographic}.

Particular care must be taken against quantum adversaries to choose an extractor that has provable security when considering potentially quantum side information. \nw{In general, quantum-secure randomness extraction can be seen as a function $\mathrm{Ext}: \{0,1\}^{h} \times \{0,1\}^d\rightarrow \{0,1\}^l$ that involves processing a block of size $h=mb$  (the $m$, $b$-bit measurement outcomes) along with a random $d$-bit seed to produce an $l$-bit output that is $\epsilon_{l}$-close to being perfectly random.}

A very attractive choice is two-universal hashing\footnote{Let $X,S$ be sets of finite cardinality $|S|\leq|X|$. A family of hash functions $\{\mathcal{F}\}$ is a set of functions $f: X\rightarrow S$ and is called \emph{two-universal} if for $f$ drawn uniformly at random from $\mathcal{F}$, it holds that $\forall, (x,x') \in X$, $x\neq x'$, $\mathrm{Pr}[(f(x) = f(x')]\leq \frac{1}{|S|}$. The purpose of the random seed $d$ is to pick a function uniformly at random, hence $d=\log_2 |\mathcal{F}|$.} (or leftover hashing) which is secure against quantum adversaries \cite{renner2008security,tomamichel2011leftover} and can be implemented efficiently as it achieves an excellent trade-off between $\epsilon$ and $l$. It should be noted that this extractor still requires a perfectly random seed of length $d$ and thus any protocol that makes use of leftover hashing can technically only be regarded as a randomness expansion protocol \cite{pironio2013security,law2014quantum}. Whilst the length of the seed must be chosen proportional to $m$, it only has to be generated once and can be safely reused to hash arbitrarily many blocks, meaning that the initial random seed can be used to generate an unbounded amount of randomness. This also means that the seed can be hard-coded into the hashing device (for a further discussion and an explicit implementation, see \cite{frauchiger2013true}). Other quantum-secure methods, such as the Trevisan extractor, are more efficient in the length of the required seed. However, this is a more computationally expensive process and cannot currently be performed at speeds at which our protocol can generate raw randomness. Thus, to achieve bit-generation rates of the same speed as the randomness generation rates reported here, it seems necessary to perform randomness extraction via leftover hashing.

We now have the tools to write down the following result for certified randomness expansion. Although this is essentially a repeat of standard techniques (see e.g. \cite{tomamichel2011leftover,frauchiger2013true}) adapted to our specific setup, we state it as a standalone theorem for completeness.

\begin{Theorem}
A certified SDI (m,$\kappa$,$\epsilon_{\mathrm{fail,m}}$,$\epsilon_c$)-randomness generation protocol as defined in Definition \ref{QRGdef} can be processed with a randomness generation seed of length $m$ via leftover hashing to produce a certified SDI random string of length
\eqn{l \ee \kappa + 2 -\log_2 \frac{1}{\epsilon_{\mathrm{hash}}^2}\label{lexp} \,,}
that is $\epsilon_c$-complete and $\epsilon_{\mathrm{hash}}+\epsilon_{\mathrm{fail,m}}$ secure.
\end{Theorem}

\begin{proof}

{\bf Security:} Let $X$ be the variable describing the measurement outcomes. Recall that the output of the randomness generation protocol after the measurement including the potential side information can be written as a classical-quantum state
\eqn{\hat{\rho}_{XE} = \sum_{x\in\mathcal{X}} P_X(x) \ket{x}\bra{x}\otimes \hat{\rho}_E^x \label{cq} \,,}
where $\mathcal{X}$ is the alphabet of possible measurement outcomes and $\hat{\rho}_E^x$ is the state of the eavesdropper given the outcome $x$. A randomly chosen leftover hashing function is then applied to distill a final random string denoted by the variable $S$. The joint state is now
\eqn{\hat{\rho}_{SE} = \sum_{s} P_{S}(s) \ket{s}\bra{s}\otimes \hat{\rho}_E^{s}\,.}

We then apply the Leftover Hash Lemma with quantum side information \cite{tomamichel2011leftover} and its extension to infinite dimensional Hilbert spaces \cite{berta2016smooth, furrer2014position} which is necessary for our purposes.

\begin{Lemma} \label{leftover}
Let $\hat{\rho}_{XE}$ be a state of the form in Eq.~(\ref{cq}) where $X$ is defined over a discrete-valued and finite alphabet and E is a finite or infinite dimensional system. If one applies a hashing function drawn at random from a family of two-universal hash functions that maps $X$ to $S$ and generates a string of length $\it{l}$, then
\nw{\eqn{D(\hat{\rho}_{SE}, \hat{\tau}_{S}\otimes \hat{\sigma}_E) \leq 2^{\frac{l - H_{\mathrm{min}}(X|E)-2}{2}}\label{hash}\,,}}
where $H_{\mathrm{min}}(X|E)$ is the conditional smooth min-entropy (with smoothing parameter $\epsilon = 0$) of the raw measurement data given Eve's quantum system.
\end{Lemma}
Comparing the security definitions in Eq.~(\ref{sec}) and Eq.~(\ref{hash}), we note that with an appropriate choice of $l$, we can ensure the security condition is met. In particular, we see that the smooth min-entropy is a lower bound on the extractable key length. To get a more exact expression, first notice that if we choose 
\eqn{l \ee \hmin(X|E) + 2 -2 \log_2 \left(\frac{p_{\mathrm{pass}}}{\epsilon_{\mathrm{hash}}}\right) \,,} 
for some $\epsilon_{\mathrm{hash}}>0$, then the right hand side of Eq.~(\ref{hash}) becomes $\epsilon_{\mathrm{hash}}/p_{\mathrm{pass}}$. Then, provided we have definitively bounded the smooth min-entropy, we will satisfy Eq.~(\ref{sec}) for any $\epsilon_{\mathrm{hash}}>0$. Finally since $\log_2(p_{\mathrm{pass}}) <0$, we have 
\eqn{l \geq \hmin(X|E) + 2 -\log_2 \left( \frac{1}{\epsilon_{\mathrm{hash}}^2}\right) \,.}

Now, suppose that we are only able to bound joint probability of passing the test whilst outputting a small smooth min-entropy $\hmin(X|E)< \kappa$ with a certain probability $\epsilon_{\mathrm{fail,m}}$ as is the case here. Then, the convexity and boundedness of the trace distance implies that this string of length $l$ will be $\epsilon_l$-secure for any security parameter
\eqn{\epsilon_l \geq \epsilon_{\mathrm{hash}}+ \epsilon_{\mathrm{fail,m}}\label{ese2} \,,}
if the length is chosen as per Eq.~(\ref{lexp}). 

{\bf Completeness:} This follows immediately from the completeness of the certified randomness generation protocol.

\end{proof}

\section{\nw{Experimental details for the real-time extraction of certified quantum random numbers}} 
\label{sec:FPGA_details}
\nw{In order to generate certified random numbers in real-time, the post-processing was implemented with a high-performance FPGA (Zynq Ultrascale$+$ ZU9EG) installed on the commercially available Printed Circuit Board (PCB) Zynq UltraScale$+$ MPSoC ZCU102 evaluation kit as shown in Fig.~\ref{fig:FPGA_Schematics}. For data acquisition, a 12-bit ADC (Analog Devices AD9625) is used while being installed on a separate PCB connected to the FPGA via an FPGA Mezzanine Card (FMC) as can be seen in the inset to Fig.~\ref{fig:FPGA_Schematics}. The evaluation kit provides several modules for data transmission, including the cage for Small Form-factor Pluggable (SFP) modules and a Universal Serial Bus (USB) 3.0 port. The Double Data Rate 4th Generation Random Access Memory (DDR4 RAM) module required for data testing is also included.



\begin{figure}[h!]
\boxed{\includegraphics[width=\columnwidth]{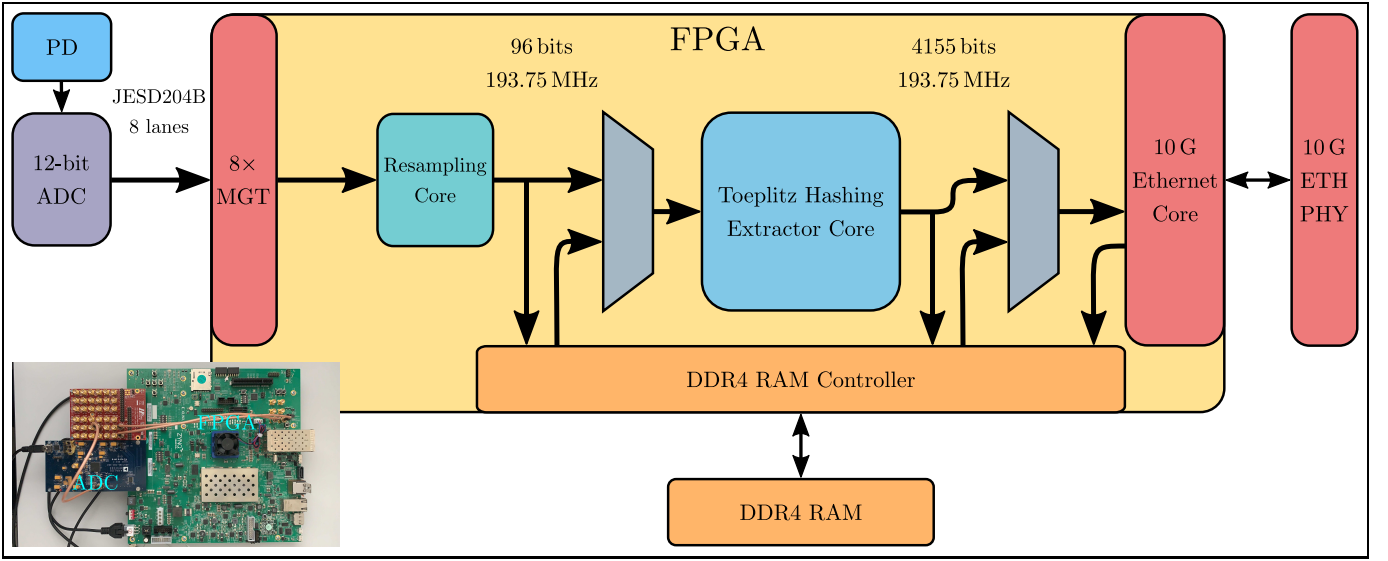}}
\caption{\nw{Schematic of the real-time post-processing board used to generate certified random numbers. The analog signal generated by the optical setup described in Fig.~\ref{fig:optical_setup} is digitised by an ADC and then further processed by an FPGA board. Additionally, the number of bits during each step of the process, along with the inverse duration of each time step, is shown above the various modules in the schematic. PD: photodiode; ADC: analog-to-digital converter; MGT: multi-gigabit transceiver; DDR4 RAM: double data rate 4th generation random access memory; ETH PHY: ethernet physical layer. Inset: Photograph of the actual post-processing board comprising the ADC and the FPGA.}}
\label{fig:FPGA_Schematics}
\end{figure}


The process described by Fig.~\ref{fig:FPGA_Schematics} is summarised as follows. The data from the ADC is deserialised with 8 Multi-Gigabit Transceivers ($8\times$ MGT) and reaches the resampling core of the FPGA where it is resampled to a lower frequency of $1.55\,$GS/s since the ADC's sampling rate is larger than the experiment's data generation (imposed by the balanced detector's bandwidth). Then, the data arrives at a multiplexing unit (grey parallelogram) followed by the central Toeplitz hashing module. Toeplitz hashing is realised via the concurrent pipeline algorithm (detailed in \cite{zhang2016fpga}) with a clock rate of $R_{\mathrm{hash}} = 193.75\,\si{\mega\hertz}$. Here, a $9600\times4155$ random Toeplitz matrix initially saved in the FPGA's memory is utilised. Indeed, it is proven in Appendix A of \cite{frauchiger2013true} that one need not renew the random input seed used to construct the Toeplitz matrix. Furthermore, for optimisation purposes, the initial large Toeplitz matrix is evenly decomposed into a series of submatrices which are multiplied sequentially with the raw input data. These submatrices have sizes of $96 \times 4155$, where $k=96\,$bits is carefully chosen to be a multiple of both the ADC's bit-depth $b=12\,$bits and the hashing block size $h=9600\,$bits. Note that the submatrix' number of rows also corresponds to the precise number of bits injected into the FPGA board per time step of the hashing algorithm, i.e. $k = \frac{12\times 1.55\times 10^{9}}{193.75 \times 10^{6}}=96$. As a result of this, substrings of $96\,$bits from the raw data at each time step are extracted and then multiplied with a corresponding random $96\times4155$ Toeplitz submatrix, thereby obtaining a single substring of $l=4155\,$bits per clock period. The XOR operation required between pairs of such subsequent strings of $4155\,$bits is performed concurrently with multiplication steps. The multiplication of the entire large Toeplitz matrix with the raw random string of $9600\,$bits is thus performed over $9600/96 = 100$ time steps, leading to an overall extraction of $4155\,$bits for every such procedure labelled as a single \textit{extraction period}. Finally, while the following extraction period commences, the previously obtained block of hashed data is prepared for the final output. 


For validation and debugging purposes, the option of saving both raw and hashed data in the FPGA's memory is implemented such that one may extract them for further analysis on a PC. Conversely, data can be uploaded to the FPGA's memory from an external source (e.g. from an oscilloscope’s ADC) and then processed by the Toeplitz hashing extractor in the FPGA. 

} 

\section{Rate comparison with homodyne protocols \label{comp}}
In this appendix, we will derive the curves shown in Fig.~\ref{fig:comparison_our_work_vs_homodyne} which compare the rates for this work to those for the device-dependent homodyning and the semi-SDI protocols with certification based on an entropic uncertainty relation \cite{marangon2017source,avesani2018source,michel2019real}. Strictly speaking, direct comparison with the EUR protocols is impossible since these fail to give a composable security parameter. Also, in practice, the achievable rates depend heavily on many technical constraints such as the detector noise and especially the number of ADC bits. Consequently, we consider a simpler, idealised calculation of the ultimate rates of these different protocols and identify fundamentally different scalings in some instances. Specifically, we will calculate the expected value of the amount of min-entropy generated per round.

\subsection{Device-dependent homodyning}
Following Haw et al \cite{Haw:2015kx}, we can upper bound the min-entropy by noting that for arbitrarily many ADC bits and perfect photon number resolving detectors, the probability distribution of the photon difference is only resolution-limited by the photon-counting measurement itself and the amplitude of the local oscillator. Specifically, it is straightforward to show that the photon difference for an arbitrary input signal mode mixed on a 50:50 beamsplitter with a coherent state $\ket{\alpha_{\mathrm{LO}}}$ gives output modes $\anih_1 = (\anih_s + \anih_{\mathrm{LO}})/\sqrt{2}$ and $\anih_2 = (\anih_s - \anih_{\mathrm{LO}})/\sqrt{2}$. The photon difference is then given by 
\eqn{\hat{I}:=\adag_1\hat{a}_1 - \adag_2\hat{a}_2 = \adag_{\mathrm{LO}}\anih_s + \adag_s\anih_{\mathrm{LO}}\,.} 

If the LO is very bright, then we can know its quadrature displacement up to an uncertainty that is very small relative to the displacement's mean. Moreover if the LO is very large relative to the photon number of the input signal, this signal will be very close to a quadrature measurement of the input signal. Following e.g. \cite{bachor}, one way to see this is to consider a decomposition of the LO operator $\hat{a}_{\mathrm{LO}} = \alpha_{\mathrm{LO}} + \delta\hat{A}_{\mathrm{LO}}$, where $\alpha_{\mathrm{LO}} $ is the mean value and the operator and $\delta\hat{A}_{\mathrm{LO}}$ represents the quantum fluctuations. Taking $\alpha_{\mathrm{LO}}$ to be real, we have
\eqn{\hat{I} = \alpha_{\mathrm{LO}} \hat{x}_s + \delta\hat{A}^\dag_{\mathrm{LO}}\anih_s + \anih_s\delta\hat{A}_{\mathrm{LO}}\label{quadflux}\,.}

If the mean LO amplitude is large with respect to fluctuations and the amplitude of the signal mode, then one has $\hat{I} \approx \alpha_{\mathrm{LO}}\hat{x}_s$. In the case of ideal detectors that can distinguish between $n$ and $n+1$ photons, this is equivalent to measuring the input quadrature with a resolution given by $\Delta = 1/\alpha_{\mathrm{LO}}$ (i.e. the rescaling by the LO power). One can also calculate the variance for an arbitrary signal state $\hat{\rho}_{s}$ with a coherent state as the LO. Defining the appropriate expectation value as $\langle\hat{I}\rangle_{\alpha_{\mathrm{LO}}} =  \mathrm{tr}\left \{ \hat{I} \bk{ \hat{\rho}_{s}\otimes\ket{\alpha_{\mathrm{LO}}}\bra{\alpha_{\mathrm{LO}}} }\right \}$, we have
\eqn{\mathrm{Var}(\hat{I}) \ee \langle\hat{I}^2\rangle_{\alpha_{\mathrm{LO}}} -  \langle\hat{I}\rangle^2_{\alpha_{\mathrm{LO}}} \nn\\
\ee \langle \alpha_{\mathrm{LO}}^{*2}\anih_s^2 + \hat{n}_{\mathrm{LO}}\anih_s\adag_s + \hat{n}_s \anih_{\mathrm{LO}}\adag_{\mathrm{LO}} + \alpha_{\mathrm{LO}}^2 \anih_s^{\dag2} \rangle - \alpha_{\mathrm{LO}}^{*2}\langle\anih_s\rangle^2 - \alpha_{\mathrm{LO}}\langle\anih_s\rangle^2\nn\\
\ee \alpha_{\mathrm{LO}}^{2} \bk{\langle\anih_s^2 + \anih_s^{\dag2} \rangle +1 +2n_s}+ n_s - \alpha_{\mathrm{LO}}^{2}\langle \hat{x}_s \rangle^2 \nn\\
 \ee n_{\mathrm{LO}} \mathrm{Var}(\hat{x}_s) + n_s \,, \label{var}} 
where we have again taken $\alpha_{\mathrm{LO}}$ to be real.

\subsubsection{Vacuum input}
In the device-dependent case where the signal is known to be vacuum, the rescaled output is a discretised Gaussian distribution with variance $V=1$ and zero mean. If we label the discretised output with index $k$, the probability distribution from the perspective of an eavesdropper (here there is no technical noise) is given by
\eqn{p(k|E) = \frac{1}{2} \left(\mathrm{erf}\bk{\frac{k\Delta +\Delta/2}{\sqrt{2V}}} -  \mathrm{erf}\bk{\frac{k\Delta -\Delta/2}{\sqrt{2V}}} \right) \,, \label{pkapprox}}
where $k \in \{ 0, \pm1,\pm2,... \}$. 

For small $\Delta$ relative to $V$, Eq.~(\ref{pkapprox}) is well approximated by
\eqn{p(k|E) = \frac{\Delta}{\sqrt{2\pi V} }\exp\bk{-\frac{(k\Delta)^2}{2V}} \,, \label{papprox}}
and the min-entropy $\hmin^{\mathrm{DD}}(X|E) = \max_k\{-\log2(p(k|E))\}$ can be directly calculated to be \cite{Haw:2015kx}
\eqn{\hmin^{\mathrm{DD}}(X_{\mathrm{vac}}|E) = \half\log_2\bk{2\pi \alpha_{\mathrm{LO}}^2} =  \half\log_2\bk{2\pi n_{\mathrm{LO}} } \,, \label{ddvac}}
where $n_{\mathrm{LO}}$ is the mean photon number present in the LO. 

\subsubsection{Coherent state input}
This rate as calculated via Eq.~(\ref{papprox}) is also unchanged if the vacuum is replaced by a coherent state since the variance of coherent states is still unity. However, if the signal is a large coherent state $\ket{\alpha_{s}}$, the approximations we utilised to derive Eq.~(\ref{papprox}) no longer hold. The other term in Eq.~(\ref{quadflux}) will not remain negligible and the fluctuations will actually increase. Considering the photon detections directly, the state after the beamsplitter will now be $\ket{\frac{\alpha_{\mathrm{LO}} + \alpha_s}{\sqrt{2}}}\otimes\ket{\frac{\alpha_{\mathrm{LO}}+\alpha_s}{\sqrt{2}}}$. The output at each detector would be described by a Poissonian distribution, which for large photon number will be well approximated by a Gaussian distribution, as will the photon difference. The variance is straightforwardly calculated to be
\eqn{V_{\mathrm{coh}} = |\alpha_{\mathrm{LO}}|^2 + |\alpha_s|^2 \,, \label{vcoh}} 
from which we can immediately read off the min-entropy as
\eqn{\hmin^{\mathrm{DD}}(X_{\mathrm{coh}}|E)  =  \half\log_2\bk{2\pi \bk{n_s+ n_{\mathrm{LO}}} } \,. \label{ddcoh} }

\subsubsection{Thermal state input}
On the other hand, if the vacuum source was instead replaced by Eve with one half of an entangled two-mode squeezed vacuum (TMSV) state
\eqn{\ket{\mathrm{TMSV}} = \frac{1}{\cosh(r)} \sum_{n=0}^\infty \left(-\tanh{r}\right)^n\ket{n,n}\,,}
then the input to the randomness measurement will be a thermal state with mean photon number $\bar{n} = \sinh^2(r)$ and quadrature variance $V = 2\bar{n}  +1$. As the amount of squeezing --- and hence the number of photons in the input state --- increases, the quadrature measurements will start to become more and more predictable and the min-entropy will decrease. Eventually, however, for a sufficiently bright TMSV state, the extra terms in Eq.~(\ref{quadflux}) become non-negligible and extra fluctuations will arise such that the overall entropy will begin to increase again. For all levels of squeezing, the statistics will be well-approximated as being Gaussian.

We can get an upper bound for the device-dependent min-entropy by assuming that Eve makes an $\hat{x}$ quadrature measurement on her half of the TMSV state. This would project the other arm into a $\hat{x}$-squeezed coherent state with variance $V_x = \mathrm{sech}(2r) = \mathrm{sech}(2(\sinh^{-1}(\sqrt{\bar{n}})))$ and a displacement given by $\sqrt{1-1/V_x^2}x_E$, where $x_E$ is the outcome of Eve's measurement. We can write down Eve's conditional guessing probability directly since it would simply be the same kind of coarse grained Gaussian distribution as before with a resolution of $1/\alpha$, but now the variance given by evaluating Eq.~(\ref{var}) to obtain
\eqn{V_{\mathrm{th}} \ee n_{\mathrm{LO}}  V_x+ \bar{n} \nn\\
\ee n_{\mathrm{LO}}  \mathrm{sech}(2(\sinh^{-1}(\sqrt{\bar{n}}))) + \bar{n}\,. \label{vth}} 

The min-entropy is then given by substitution in Eq.~(\ref{pkapprox}), leading to
\eqn{\hmin^{\mathrm{DD}}(X_{\mathrm{th}}|E) \ee \half\log_2\bk{2\pi (n_{\mathrm{LO}}V_x + \bar{n} )} \,. \label{ddtherm}}

Note that this is an upper bound because we are calculating the min-entropy that Eve would have about an individual round of the protocol. In theory, in a protocol where Eve's goal was to guess the $n$-symbol output of an $n$-round protocol, she could potentially employ a collective measurement that might further reduce her uncertainty. Nevertheless, we will proceed with this device-dependent upper bound for comparative purposes.

\subsection{Entropic uncertainty relation certified homodyning}
In the works \cite{marangon2017source, avesani2018source, michel2019real}, the randomness present in the $X$ quadrature is certified by making measurements in the conjugate $P$ quadrature basis and exploiting an entropic uncertainty relation of the form
\eqn{\hmin^{\mathrm{EUR}}(X|E) \geq \log_2(c) - \hmax(P|B) \,, \label{eur_equation}}
where $\hmax(X) = 2\log_2\bk{\sum_x \sqrt{p_x}}$. 

In fact, to get the expected value for the min-entropy generation rate, one should multiply the right-hand side of Eq.~(\ref{eur_equation}) by the probability $p_X$ that a round is used as a randomness generation round rather than a check round, and also subtract some randomness used to randomly switch bases in the future iterations of the protocol \cite{vallone2014quantum,marangon2017source}. Here, we will set $p_X = 0.1$ as per \cite{michel2019real} and to get an upper bound for comparison purposes, we will ignore the random seed term.
For discretised homodyne measurements (assuming symmetric quadrature resolution $\Delta$), one has that $c = \frac{2\pi}
{\Delta^2}$ and noting that $\hmax(P|B)\leq\hmax(P)$, we get
\eqn{\hmin^{\mathrm{EUR}}(X|E)&\geq& p_X\bk{ \log_2\bk{\frac{2\pi}{\Delta^2}} - \hmax(P) }\nn\\
\ee p_X \bk{\log_2\bk{2\pi n_{\mathrm{LO}}} - \hmax\bk{P}} \,.} 

Using the Jacobi theta functions $\vartheta_3(z,\tau) = \sum_{n=-\infty}^{\infty} \tau^{n^2}e^{ 2 n i z} $, we can rewrite Eq.~(\ref{papprox}) to directly evaluate the max-entropy to find
\eqn{\hmin^{\mathrm{EUR}}(X|E)&\geq&p_X \left ( \log_2\bk{2\pi n_{\mathrm{LO}}} - \log_2\bk{\frac{\vartheta_3\left(0,e^{-1/(4n_{\mathrm{LO}} V)}\right)^2}{\sqrt{2\pi n_{\mathrm{LO}}V} } } \right) \,. \label{eur}}

Using this formula, we can evaluate the EUR-based certified randomness rates for the variance appropriate for each input state; namely the coherent and thermal cases exposed in Eq.~(\ref{vcoh}) and Eq.~(\ref{vth}), respectively.

Note that this rate represents an over-estimation of the randomness generated in that we are using the max-entropy exactly. In practice, this would have to be estimated from statistics (see \cite{michel2019real} for several estimators) which would generally result in a lower value for the certified min-entropy.

\subsection{This work}
Here, we compare the device-dependent and EUR-based rates with our work. In fact, the EUR-based rates cannot be directly compared because in reality entropic terms should be empirically bounded in a way that gives composable $\epsilon$-security (i.e. there is a test such that the joint probability of passing the test whilst having less than the certified rate should be less than $\epsilon$). For this idealised calculation, our rates are given by Theorem \ref{rgendideal}. Recall that our protocol is probabilistic, meaning that randomness is only certified when the test is passed by observing $n_C^-$ or more photons in the certification measurement, which will happen with a probability at least $1-\epsilon_c$. From Theorem \ref{rgendideal}, we know that either the test will fail or the min-entropy will be strictly lower bounded as per Eq.~(\ref{hminthmideal}). Putting all of this together, we can say that the expected min-entropy generated in a single round (i.e. $m=1$) will be 
\eqn{\hmin^{\mathrm{SDI}}(X|E) \geq (1-\epsilon_c) \left(\frac{1}{2} \log_2\bk{\half \pi n_{R}^{-}} - \mathcal{O}\bk{\frac{1}{n_R^-}} \right) \,, \label{ourhmin}}
with a failure parameter of 
\eqn{\epsilon_{\mathrm{fail}} = \exp\bk{-\frac{2 (r_{1}  (n_{R}^{-}+n_{C}^{-}-1)-n_{C}^{-})^2}{n_{R}^{-}+n_{C}^{-}-1}} \,. \label{epsilon1}}

Notice that for the regions of interest in Fig.~\ref{fig:comparison_our_work_vs_homodyne}, namely where this curve surpasses the EUR curves and scales similarly to the device-dependent case, the inferred photon number will be such that the corrective term $\mathcal{O}(1/n_R^-)$ is negligible. To evaluate this expected min-entropy given a target value for $\epsilon_{\mathrm{fail}}$ associated with the input states above, we simply need to calculate what $1-\epsilon_c$ will be for a given threshold $n_C^-$. With those in hand, we can solve Eq.~(\ref{epsilon1}) for the value of $n_R^-$ that achieves the target $\epsilon_{\mathrm{fail}}$ and then calculate the corresponding min-entropy via Eq.~(\ref{ourhmin}). 

For a coherent state input $\ket{\alpha_{s}}$, the state going into the certification measurement will be $\ket{\sqrt{r_1}\alpha_s}$. For large $\alpha_s$, the Poissonian photon-number distribution will be well approximated by a Gaussian distribution and the probability of observing $n_C^-$ or more photons will be given by $1-\epsilon_c = \frac{1}{2} \left(\text{erf}\left(\frac{\bar{n}_C-n_C^-}{\sqrt{2 \bar{n}_C} }\right)+1\right)$, where $\bar{n}_{C} = r_1\bar{n}$, with $\bar{n} = |\alpha_{s}|^2$ the mean photon number of the incoming coherent state.

Similarly, for a thermal state source, the input to the certification measurement will be a thermal state with mean photon number $\bar{n}_C = r_1n_{\mathrm{th}}$, with $n_{\mathrm{th}}$ the mean photon number of the incoming thermal state. Finally, using the formula for a geometric series and the photon number representation of a thermal state, the relationship between the threshold and the passing probability is given by $1-\epsilon_c = 1-\left(1-\left(\frac{\bar{n}_C}{\bar{n}_C+1}\right)^{n_C^{-}-1}\right)$.

\bibliography{main.bib}

\end{document}